\documentclass[11pt]{article}
\usepackage{amsmath,amssymb,amsfonts,amsthm, epsfig}
\usepackage[english]{babel}
\usepackage{mathtools} 
\usepackage[noadjust]{cite} 
\usepackage{enumitem}
\usepackage{fourier}
\usepackage[unicode]{hyperref}

\usepackage{a4wide}

\DeclareMathAlphabet{\mathcal}{OMS}{cmsy}{m}{n}

\usepackage{graphicx}
\usepackage[usenames,svgnames]{xcolor}
\usepackage{colortbl}
\usepackage{subfigure}
\hypersetup{colorlinks=true,citecolor=DarkGreen,linkcolor=DarkGreen}

\newcommand{\R}{\mathbb{R}}

\numberwithin{equation}{section}

\newtheorem{theorem}{Theorem}[section]
\newtheorem{lemma}[theorem]{Lemma}
\newtheorem{proposition}[theorem]{Proposition}

\newtheorem{remarkth}[theorem]{Remark}

\newenvironment{remark}{\begin{remarkth}\upshape}{\hfill$\diamond$\end{remarkth}}
\renewcommand{\emph}[1]{{\bfseries\itshape{#1}}}

\theoremstyle{definition}
\newtheorem{definition}{Definition}

\makeatletter
\newcommand\tocinarticle{%
     \begin{flushleft}
     \end{flushleft}
\@starttoc{toc}%
   }
\makeatother

\let\OLDthebibliography\thebibliography
\renewcommand\thebibliography[1]{
  \OLDthebibliography{#1}
  \setlength{\parskip}{0pt}
  \setlength{\itemsep}{4pt plus 0.3ex}
}

\title{\bf Reduction and relative equilibria for the 2-body problem on spaces of constant curvature}

\author{A.V.~Borisov, L.C.~Garc\'ia-Naranjo, I.S.~Mamaev \& J.~Montaldi}

\begin{document}

\maketitle

\begin{abstract}
We consider the two-body problem on surfaces of constant non-zero curvature and classify the relative equilibria and their stability.    On the hyperbolic plane, for each $q>0$ we show there are two relative equilibria where the masses are separated by a distance $q$. One of these is geometrically of elliptic type and the other of hyperbolic type.  The hyperbolic ones are always unstable, while the elliptic ones are stable when sufficiently close, but unstable when far apart.  

On the sphere of positive curvature, if the masses are different, there is a unique relative equilibrium (RE) for every angular separation except $\pi/2$. When the angle is acute, the RE is elliptic, and when it is obtuse the RE can be either elliptic or linearly unstable. We show using a KAM argument that the acute ones are almost always nonlinearly stable.  If the masses are equal there are two families of relative equilibria: one where the masses are at equal angles with the axis of rotation (`isosceles RE') and the other when the two masses subtend a right angle at the  centre of the sphere.   The isosceles RE are
 elliptic if  the angle subtended by the particles is acute and is unstable if it is obtuse. At  $\pi/2$, the two families meet and a 
 pitchfork bifurcation takes place. Right-angled RE are elliptic away from the bifurcation point.

In each of the two geometric settings, we use a global reduction to eliminate the group of symmetries and analyse the resulting reduced equations which live on a 5-dimensional phase space and possess one Casimir function.  
\end{abstract}

\setcounter{tocdepth}{2}
\tocinarticle


\section*{Introduction}
\addcontentsline{toc}{section}{Introduction}

The study of the dynamics of material points and rigid bodies in  spaces of constant curvature has been a popular subject of research in the past two decades.
For recent advances in this area, see the review~\cite{1, 7} and the book by Diacu~\cite{Diacu} (and its critique in \cite{1}).

In this paper we focus on the 2-body problem on a (complete and simply connected)  two dimensional space of constant non-zero curvature. 
 Our interest in the problem is mathematical, although there is a possible physical motivation in that the background curvature of the universe may be non-zero .  However the effect of the curvature would probably be negligible at the length scales involved in isolated 2-body problems. 
Contrary to the situation  in flat space,  the system is not equivalent
to the corresponding generalisation of the Kepler problem: it is nonintegrable and exhibits chaotic behaviour.  This is due to there being no analogue of the centre of mass frame in a curved space, and no Galilean invariance. 
Recent papers like~\cite{1, 2, 3, 4,Kilin, Sh, LGN2016} have considered the reduction by symmetries and some qualitative aspects of the problem.
Libration points and choreographies are treated in~\cite{Kilin, Mont, 3} and the restricted  two-body problem is considered in~\cite{BM_2006, Ch}.

In the opinion of the authors,  there is no coherent, systematic, and complete treatment of the classification and stability of the relative equilibria (RE) of 
the problem, accessible to the community of celestial mechanics,  and the principal aim of this paper is to fill this gap.

Our approach to studying the RE of the problem relies on the use of the explicit form of the reduced equations of motion.  This
  allows us to recover  previous results  in a systematic and elementary way, and to extend them. Original  results of our paper include 
  the classification of RE for arbitrary attractive potentials, a detailed discussion of the qualitative features of the motion and the stability analysis of RE
 in the case of positive curvature, and the presentation of the energy-momentum bifurcation diagram  in the case of negative curvature.

We describe the results of the paper  in  more detail below, but we begin by recalling the general notion of relative equilibrium, which goes back to Poincar\'e and applies to any dynamical system with a continuous  group of symmetries.

\begin{definition}\label{def:RE}
Consider a dynamical system with a symmetry group $G$.  A \emph{relative equilibrium} is a trajectory of the dynamical system which coincides with the motion given by the action of a 1-parameter subgroup of the group $G$.
\end{definition}

In other words, a trajectory $\gamma(t)$ is a relative equilibrium if there is a 1-parameter subgroup $g(t)$ of $G$ such that $\gamma(t)=g(t)\cdot\gamma(0)$, and one can show that this is equivalent to the trajectory lying in a single group orbit in the phase space.  In particular, relative equilibria correspond to equilibrium points of the reduced equations of motion.
Note that if $\gamma(t)$ is a relative equilibrium, then so is $k\cdot\gamma(t)$ (for any $k\in G$), with corresponding 1-parameter subgroup $kg(t)k^{-1}$.   In the literature, the term relative equilibrium can refer to the trajectory, to any point on the trajectory, or to the entire group orbit containing the trajectory.  More details on relative equilibria in the Hamiltonian context can be found in many places, for example the Lecture Notes by Marsden \cite{Marsden}.  The symmetry groups we use in this paper are  $SO(3)$, the group of isometries of the sphere, and $SO(2,1)$ (or $SL(2,\R)$), the isometries of the Lobachevsky plane (also called the hyperbolic plane), and we describe their 1-parameter subgroups in the relevant section.  See also \cite{Mo-Peyresq,MN-G} for details about relative equilibria for these groups. 

In the problems we consider, the RE are always `rigid motions'. On the sphere, these consist simply of uniform rotations about a fixed axis.  However, in hyperbolic geometry there is more than one type of rigid motion: the so-called elliptic, hyperbolic and parabolic motions.  The elliptic motions are periodic, while the others are unbounded; see the discussion below for more details.

Another important consideration in these systems is \textit{time-reversal symmetry}, which holds whenever the Hamiltonian is given by the sum of a quadratic kinetic energy and a potential energy which depends only on the configuration. In canonical coordinates, this symmetry is given by the map on phase space $(q,p)\mapsto (q,-p)$.   If $(q(t),p(t))$ is a trajectory of the system, then its `reversal' $(q(-t), -p(-t))$ is another trajectory of the system.  In particular, if such a trajectory is a relative equilibrium then so is its reversal.  

In the statements of existence of relative equilibria, we do not distinguish between 2 trajectories that either lie in the same group orbit, or are related by time reversal.  Thus we count RE `up to all symmetries', including time reversal and exchange of the bodies when the masses are equal.

\subsubsection*{Reduction}

As mentioned above, the reduction of the problem has been considered before \cite{4,3}. We have nevertheless included a self-contained presentation of the reduction
for completeness.

For both the positive and negative curvature cases, the unreduced system is a four degree of freedom symplectic Hamiltonian system,  the symmetry group is
three dimensional and acts freely and properly.  The reduced system is  a five dimensional Poisson Hamiltonian system, whose generic symplectic
leaves are the four dimensional level sets of a Casimir function.

We first deal with the case of positive curvature and  consider the problem on the $2$-dimensional sphere $S^2$.
We perform the reduction of the problem by the action of $SO(3)$ that simultaneously rotates both masses.  In our treatment we do not consider collisions nor antipodal configurations, which allows us to   introduce a moving coordinate frame whose axes are aligned  according
to the configuration of the masses in a convenient way. The Hamiltonian of the system may then be written in terms of 
the angle $q$ subtended by the masses at the centre of the sphere, its  conjugate momentum $p$, and the
 vector of angular momentum $\boldsymbol m$ written in  the moving frame. These quantities do not depend on the orientation of the fixed frame and may therefore
 be used  as coordinates on the
reduced space. This approach is inspired by the reduction of the free rigid body problem and, just as   happens for that problem, 
the Euclidean squared norm of $\boldsymbol m$ passes down to the quotient space as a Casimir function whose level sets are the symplectic leaves of the reduced space.

We apply an analogous reduction scheme in the case of negative curvature by considering the action of $SO(2,1)$ (or equivalently of $SL(2,\R)$) on the Lobachevsky plane $L^2$.  
This time, the Casimir function $C$ on the reduced space is the squared norm of the momentum vector  $\boldsymbol m$ with respect to the Minkowski metric.

\subsection*{Classification and stability of relative equilibria}

After finding the reduced equations, in the final two sections we proceed to classify the RE of the problem by finding all the equilibria of the reduced equations.  In this way we recover the results of \cite{5,DiacuPCh,LGN2016} in a systematic and elementary fashion. Moreover, with this approach, we are able to conveniently analyse their stability.  We now summarize these results.

\subsubsection*{The case of negative curvature}

In hyperbolic geometry, there are 3 types of non-trivial isometry, known as elliptic, hyperbolic and parabolic transformations and correspondingly 3 types of 1-parameter subgroup of $SO(2,1)$.  The elliptic transformations are characterised by having precisely one fixed point in the hyperbolic plane.  The hyperbolic and parabolic transformations have no fixed points, but have respectively two and one fixed points `at infinity'; see any text on hyperbolic geometry, for example \cite{Iversen}. 

As is known \cite{DiacuPCh,LGN2016}, for the 2-body problem on the Lobachevsky plane $L^2$, there are two families of RE, known as the \emph{hyperbolic} and \emph{elliptic} families, according to the type of their 1-parameter subgroup.  The former are unbounded solutions that do not have an analog in Euclidean space, while the latter are periodic solutions analogous to the RE in Euclidean space. This classification becomes very transparent in  our treatment: hyperbolic RE correspond to equilibria of the reduced Hamiltonian system restricted to negative values of the Casimir function $C$, whereas elliptic RE are those for positive values of $C$ (see Section\,\ref{SS:reconstruction-L2} below).
It is also known that  parabolic transformations do not give rise to RE of the problem \cite{DiacuPCh,LGN2016}. In our treatment, this corresponds
to the absence of equilibria of the reduced system when $C=0$.

The (nonlinear) stability properties of the RE described above was first established in~\cite{LGN2016} by  working on a symplectic slice for the unreduced system since the reduced equations
were not available (in a journal in English!).    With the reduced system at hand, we are able to  recover these results in elementary terms by directly analysing the signature of the Hessian of the reduced Hamiltonian at the RE. The hyperbolic RE are always unstable, whereas the elliptic RE are stable if the masses are sufficiently close. However, as the distance between them grows, the family undergoes a saddle-node bifurcation and the elliptic RE become unstable.

All of the information of the RE of the system is illustrated in the Energy-Momentum diagram shown in Fig.\,\ref{F:DiagramL2}.
This kind of diagram, as well as the underlying topological considerations of the analysis, goes back to Smale and has been
been developed in detail for integrable systems by Bolsinov, Borisov
and Mamaev~\cite{5}. The application of this kind of analysis to  nonintegrable systems, such as the one considered in this paper, is not very common.\footnote{We mention  only the paper~\cite{6} on the Conley index where new isosceles vortex configurations were found and their stability was established using topological methods.}

\subsubsection*{The case of positive curvature}

 In this case  all RE are periodic solutions in which the masses rotate
 about an axis that passes through the shortest geodesic that joins them.
  As indicated first in \cite{3}, the classification of RE in this case is more intricate than for negative curvature since  it depends on how the masses of the bodies compare to each other:
  \begin{enumerate}
\item   If the masses are different, there are
  two disjoint families of RE that we term \emph{acute} and \emph{obtuse}, according to the (constant) value of the angle between the masses during the motion.
  For acute RE, it is the heavier mass that is closer to the axis of rotation and hence these are a natural generalisation of the RE of the problem in Euclidean space.
  On the other hand, for obtuse RE it is the lighter mass which is closer to the axis of rotation, and these RE do not have an analog in Euclidean (or hyperbolic) space.

\item If the masses are equal, there are two families of RE.  The family of \emph{isosceles} RE are those for which the axis of rotation bisects the arc that joins the two masses, while in the family of \emph{right angled} RE  the angle between the masses is $\pi/2$  and the axis of rotation is located anywhere between them. These two families meet when the axis of rotation subtends an angle of $\pi/4$ with each mass, and a pitchfork bifurcation of the RE of the system takes place.
  \end{enumerate}

As for the case of negative curvature, we  compute the signature of the Hessian of the reduced
Hamiltonian at the RE in an attempt to establish stability results. Via this analysis it is possible to
conclude the instability of certain RE of the system. However, contrary to the case of negative curvature, this is
 insufficient to prove any kind of  nonlinear stability results of the RE since the Hessian matrix is not definite and hence
 the reduced Hamiltonian may not be used as a Lyapunov function of the system. This surprising feature of the
 problem was also found in \cite{MRO-2017} by working on a symplectic slice of the unreduced system.

In view of the above considerations, we take an analytical approach to the study of the nonlinear stability of certain RE of the problem.
By using Birkhoff normal forms and applying KAM theory, we are able to show that, if the masses are different, the generic acute RE of the problem are stable.

\subsection*{Stability and reduced stability} 
Since we study here the reduced equations of motion, the stability we consider is \emph{reduced stability} (whether nonlinear or linear), and it is important to appreciate the relationship between this reduced stability and the stability in the full unreduced equations of motion.  This involves the `geometry' of the momentum value for a given relative equilibrium. 

Many of the relative equilibria are periodic motions, and for such motion there is a well-known concept of \emph{orbital stability}, where there is a tubular neighbourhood of the orbit such that any solution that intersects that neighbourhood stays close to the orbit. 

If $\boldsymbol M$ is the momentum value at (any point of) the relative equilibrium, then there is a subgroup of the symmetry group $G$ which fixes that value, which we denote $G_{\boldsymbol M}$. If this subgroup is isomorphic to $SO(2)$ (which it is in all the RE that enjoy reduced stability) then nonlinear reduced stability implies orbital stability.  
For a discussion concerning stability of RE in the Lobachevski plane, see for example \cite{MN-G}.

\subsection*{Outline of the paper}

The reduction of the problem on $S^2$ and $L^2$ is respectively presented in Sections 1 and 2. The case of  positive curvature is presented first since the geometry is better known.  We then proceed to classify the RE of the problem and study their stability. We first deal with the case 
of negative curvature in Section 3 and then with the case of positive curvature in Section 4. We have chosen to present first the negative curvature
results since, as was discussed above, the analysis is more straightforward.  Finally, some related open problems are described at the end.

\newpage
\section{Reduction for the sphere $S^2$}

The 2-body problem on $S^2$ concerns the motion of two masses $\mu_1$ and $\mu_2$ on the unit sphere on $\R^3$, that are subject to the attractive force 
that only depends on the distance between the particles. The accepted generalization of the inverse squared law from the planar case  to the problem on the sphere is defined
by the potential
\begin{equation} \label{eq:gravity on S2}
U_{\mathrm{grav}}(q)=-\frac{G\mu_1\mu_2}{\tan q},
\end{equation}
where $q\in (0,\pi)$ is the angle subtended by the masses (their Riemannian distance) and $G>0$ is the `gravitational' constant. This form of the potential leads to Bertrand's property in Kepler's problem
 and to a natural generalization
of Kepler's first law \cite{kozlov, Carinena}.

Note that the potential $U_{\mathrm{grav}}$ is singular at configurations where $q=0$ and $q=\pi$. Namely, at collisions and antipodal positions of the particles. 
In this paper we consider 
 more general attractive, $q$-dependent potentials $U:(0,\pi)\to \R$, that have the same qualitative properties as $U_{\mathrm{grav}}$:
\begin{equation*}
U'(q)>0, \qquad \lim_{q\to 0^+}U(q) =-\infty, \qquad \lim_{q\to \pi^-}U(q) =\infty.
\end{equation*}

The configuration space for the system is
\begin{equation} \label{eq:Q for sphere}
Q=(S^2\times S^2)\setminus \Delta, 
\end{equation}
where $\Delta$ is the set of collisions and antipodal configurations.  $Q$ is a  four dimensional manifold which is an open dense subset of $S^2\times S^2$.
The momentum phase space of the system is the eight-dimensional manifold $T^*Q$. The dynamics of the system is Hamiltonian with respect to the canonical
symplectic structure on $T^*Q$, and the Hamiltonian $H:T^*Q\to \R$ given by $H=T+U$ where $T$ is the sum of the kinetic energy of the particles.

The system is clearly invariant under the (cotangent lift) of the action of $SO(3)$ that simultaneously rotates both particles; this action is free and consequently the reduced orbit space $T^*Q/SO(3)$ is a manifold.

\begin{theorem}
\label{th:reduction S2}
For $Q$ given in \eqref{eq:Q for sphere}, the reduced space $T^*Q/SO(3)$ is isomorphic as a Poisson manifold to $\R^3\times (0,\pi)\times \R\ni ({\boldsymbol m},q,p)$. The Poisson structure on this space is defined by the relations
\begin{equation}\label{eq:PB on S2}
\{m_x, m_y\}=-m_z, \quad \{m_y, m_z\}=-m_x, \quad \{m_z, m_x\}=-m_y, \quad \{q, p\}=1,
\end{equation}
where  $\boldsymbol m=(m_x,m_y,m_z)^T$.

For the 2-body problem on the sphere with interaction governed by potential energy $U(q)$, the reduced dynamics is Hamiltonian with respect to this Poisson structure  
with Hamiltonian 
\begin{equation}
\label{eq:Hamiltonian on sphere}
H({\boldsymbol m},q,p)=\frac{1}{2\mu_1 }\left ( (\boldsymbol  m , {\bf A}(q) \boldsymbol  m)
+2 m_x p + ( 1 +  \mu)p^2 \right ) + U(q),
\end{equation}
where $(\cdot, \cdot)$ is the Euclidean scalar product in $\R^3$ and
\begin{equation}
\label{eq:defA}
 {\bf A}(q) =\begin{pmatrix}  1 & 0 & 0 \\ 0 & 1 &   \dfrac{\cos q}{\sin q} \\
0 &  \dfrac{  \cos q}{\sin q}  & \dfrac{ (\mu +\cos^2 q)}{ \sin^2 q} \end{pmatrix}, \qquad \mu=\frac{\mu_1}{\mu_2}.
\end{equation}
\end{theorem}

As a consequence of the above theorem, the reduced equations of motion take the form
\begin{equation}
\label{eq5.2}
\dot{\boldsymbol m}=\boldsymbol m\times\frac{\partial H}{\partial \boldsymbol m},\quad
\dot{q}=\frac{\partial H}{\partial p},\quad
\dot{p}=-\frac{\partial H}{\partial q},
\end{equation}
with $H$ given by \eqref{eq3.00}. Note that the  Poisson bracket \eqref{eq:PB on S2} is degenerate and possesses the Casimir function
\begin{equation}
\label{eq3}
C(\boldsymbol m)=(\boldsymbol m,\boldsymbol m)=m^{2}_{x}+m^{2}_{y}+m^{2}_{z},
\end{equation}
which is  a first integral of \eqref{eq5.2}.

The theorem is proved in the following 2 subsections.

\subsection{Group parametrization of configurations}	
\label{SS:Group-parametrization}

 Let $OXYZ$~be a fixed coordinate system on $\R^3$ and let $\boldsymbol R_{\alpha}=(X_{\alpha},Y_{\alpha},Z_{\alpha})^T$~be the
Cartesian coordinates of the point mass~$\mu_{\alpha}$, $\alpha=1,2$.
The key idea behind the proof of Theorem~\ref{th:reduction S2} is the introduction of a moving orthogonal coordinate system~$Oxyz$ 
 determined by the following two conditions (see Fig.\,\ref{fig1}):
\begin{enumerate}
\item[(i)]  the axis~$Oz$ passes through the mass~$\mu_1$,
\item[(ii)] the mass~$\mu_2$ is contained in the  plane~$Oyz$ with coordinate $y>0$.
\end{enumerate}
\begin{figure}[h]
\centering
\includegraphics[totalheight=5cm]{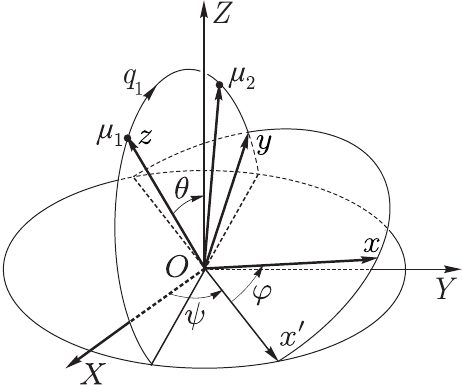}
\caption{Euler angles for the 2-body configuration on $S^2$.}
\label{fig1}
\end{figure}

According to our convention, the coordinates of the masses on the moving frame are given by the vectors
\begin{equation}
\label{eq:masses-body-frame}
\boldsymbol r_{1}=( 0, 0, 1)^T, \qquad \boldsymbol r_{2}=( 0, \sin q, \cos q)^T.
\end{equation}

Note that associated to any configuration, we can determine the angle $q\in (0,\pi)$ subtended by the masses, and an element in $SO(3)$
via conditions (i) and (ii). This process may be inverted and shows that as a manifold
\begin{equation*}
Q=SO(3)\times (0,\pi).
\end{equation*}

An element $g\in SO(3)$ changes coordinates from the body frame $Oxyz$ into the space frame.
We introduce  Euler angles  $(\theta, \varphi, \psi)$ in $SO(3)$ according to the convention $g=R^Z_\psi R^X_\theta R^Z_\varphi$ (see Fig.\,\ref{fig1}). Then
 $0<\theta<\pi, \, 0<\varphi<2\pi, \, 0<\psi<2\pi$, and 
\begin{equation}
\label{eq:g}
\begin{small}
g(\theta, \varphi, \psi)=\begin{pmatrix} \cos\varphi\cos\psi{-}\cos\theta\sin\psi\sin\varphi &
-\sin\varphi\cos\psi{-}\cos\theta\sin\psi\cos\varphi  &   \sin\theta\sin\psi \\
\cos\varphi\sin\psi{+}\cos\theta\cos\psi\sin\varphi  & -\sin\varphi\sin\psi
{+}\cos\theta\cos\psi\cos\varphi   & -\sin\theta\cos\psi \\
\sin\theta\sin\varphi   & \sin\theta\cos\varphi  & \cos\theta \end{pmatrix}.
\end{small}
\end{equation}
Therefore we may use $(\theta, \varphi, \psi, q)$ as generalized coordinates for $Q$.  In particular, the position of the masses in the fixed frame in terms of our generalized coordinates are $\boldsymbol R_{1}= g \boldsymbol r_{1}$ and $\boldsymbol R_{2}= g \boldsymbol r_{2}$.

Using these expressions one can obtain an explicit expression for the Lagrangian of the system $L=T-U$, where the kinetic
energy 
\begin{equation}
\label{eq:KinEnergy}
T=\frac{1}{2} \left (\mu_1 || \dot{ \boldsymbol R}_{1}||^2 + \mu_2 ||  \dot{ \boldsymbol R}_{2}||^2 \right ),
\end{equation}
 and the potential energy $U=U(q)$.  To exploit the symmetries  we write $T$ in terms of the (left invariant) body frame angular velocity $\boldsymbol \omega \in \R^3$,
 defined by $\widehat{ \boldsymbol \omega }=g^{-1}\dot g$ where, as usual, $\widehat{ \boldsymbol \omega }$ is the matrix
\begin{equation}
\label{eq:hatmap}
\widehat{ \boldsymbol \omega }= \begin{pmatrix} 0 & -\omega_z & \omega_y \\ \omega_z & 0 & -\omega_x \\ -\omega_y & \omega_x & 0\end{pmatrix}.
\end{equation}
Performing the algebra one finds
\begin{equation}
\label{eq:ang velocity}
\omega_x=\dot \psi \sin \theta \sin \varphi  +\dot\theta \cos \varphi, \quad \omega_y= \dot \psi \sin \theta \cos \varphi  -\dot\theta \sin \varphi,
\quad \omega_z=\dot \varphi +\dot \psi \cos \theta,
\end{equation}
and the following expression for the kinetic energy
\begin{equation}
\label{eq:KinEnergy}
T= \frac{\mu_1}{2} ({\bf A}(q)^{-1}  \boldsymbol \omega,  \boldsymbol \omega)   +\frac{\mu_2}{2} (\omega_x -  \dot q)^2 ,
\end{equation} 
where ${\bf A}(q)$ is given by \eqref{eq:defA}. The independence of the above expression on the Euler angles is due to the $SO(3)$ invariance of the system
(the invariance under rotations of the space frame).

\subsection{Generalized momenta and reduction}
\label{SS:reductionS2}

We define the generalized momenta of the system in the standard way:
\begin{equation}
\begin{gathered}
\label{eq456}
P_\theta=\frac{\partial L}{\partial\dot{\theta}},\qquad
P_\varphi=\frac{\partial L}{\partial\dot{\varphi}},\qquad
P_\psi=\frac{\partial L}{\partial\dot{\psi}}, \qquad
p=\frac{\partial L}{\partial \dot{q}},
\end{gathered}
\end{equation}
where $L=T-U$.
The canonical Poisson structure on $T^*Q$ is determined by the canonical relations
\begin{equation}
\label{eq:canonical-Poisson}
\{\theta , P_\theta\} =1, \qquad \{\varphi , P_\varphi \} =1, \qquad \{\psi , P_\psi \} =1, \qquad \{q , p\} =1,
\end{equation}
with all other brackets equal to zero.

To perform the reduction we introduce the body frame representation of the angular momentum $\boldsymbol m :=\frac{\partial T}{\partial  \boldsymbol \omega}$.
Its expression in  terms of the canonical coordinates  is
\begin{equation}
\label{eq:m in terms of p}
\boldsymbol m =\frac{1}{\sin \theta} \begin{pmatrix} 
 \sin\varphi\,(P_\psi-P_\varphi\cos\theta)+P_\theta\sin \theta \cos\varphi \\
\cos\varphi\,(P_\psi-P_\varphi\cos\theta)-P_\theta\sin \theta \sin\varphi \\
 \sin \theta P_\varphi ,
\end{pmatrix},
\end{equation}
and the kinetic energy may be expressed as
\begin{equation*}
\label{eq3.00}
T=\frac{1}{2\mu_1 }\left ( (\boldsymbol  m , {\bf A}(q) \boldsymbol  m)
+2 m_x p + ( 1 +  \mu)p^2 \right ).
\end{equation*}
This establishes the validity of \eqref{eq:Hamiltonian on sphere}. The commutation relations  \eqref{eq:PB on S2} are directly obtained from
\eqref{eq:canonical-Poisson} and \eqref{eq:m in terms of p}, and  Theorem~\ref{th:reduction S2} is proved.

\begin{remark}Geometrically, the  reduction process 
 carried above consists of working out the left trivialization of  $T^*SO(3)$ to arrive at the decomposition
\begin{equation*}
T^*Q= T^*SO(3) \times T^*(0,\pi) = SO(3) \times so(3)^*  \times (0,\pi) \times \R,
\end{equation*}
and then eliminating the $SO(3)$ component by the symmetries. The resulting bracket is the product of the Lie-Poisson bracket on $so(3)^*$ and the
canonical bracket on $T^*(0,\pi)=(0,\pi) \times \R$ as may be recognized in \eqref{eq:PB on S2}. A fuller description of this process and the reduction for the $N$-body problem in $S^2$ can be found in \cite{1}.
\end{remark}

\subsection{Conserved quantities and reconstruction}
\label{SS:reconstruction-sphere}

Assume that we are given a solution $(\boldsymbol m(t), q(t), p(t))$ to the reduced system~\eqref{eq5.2}.
We now explain how to determine the time dependence of the Euler angles.

Observe that by the rotational invariance of the problem,  the angular momentum $\boldsymbol M$
of the system written in the  fixed axes is constant along the motion:
\[
\boldsymbol M:=\sum\limits^{2}_{\alpha=1}\mu_\alpha\boldsymbol R_\alpha\times\dot{\boldsymbol R}_\alpha.
\]
We choose the fixed frame $OXYZ$ in such a way that
$\boldsymbol M\parallel OZ$. Hence, $\boldsymbol M= (0,0,M_0)^T$ with $M_0\geq 0$, and using \eqref{eq:g}
we get
\begin{equation} \label{eq:m for S2}
\boldsymbol m(t)  =g(t)^{-1}\boldsymbol M = M_0 \begin{pmatrix} \sin\theta \sin\varphi \\ \sin\theta\cos\varphi \\ \cos\theta  \end{pmatrix} .
\end{equation}
Therefore,
\begin{equation}
\label{eq:reconstruction-sphere-thetaphi}
\cos\theta=\frac{m_{z}(t)}{M_0},\quad
\tan\varphi=\frac{m_{x}(t)}{m_{y}(t)}.
\end{equation}
A quadrature for the evolution of~$\psi$ is obtained by using~\eqref{eq:ang velocity} to write
\reversemarginpar
\begin{equation}
\begin{gathered}
\label{eq1.9e}
\dot{\psi}=\frac{\sin\varphi \, \omega_x+\cos\varphi\, \omega_y}{\sin\theta}=
\frac{M_0\left(m_{x}(t)\omega_{x}(t)+
m_{y}(t)\omega_{y}(t)\right)}{M^{2}_{0}-m_{z}^2(t)},\\
\end{gathered}
\end{equation}
where 
\begin{equation*}
\omega_x=\frac{\partial H}{\partial m_x},\quad
\omega_y=\frac{\partial H}{\partial m_y}.
\end{equation*}

A particular case of reconstruction is from relative equilibria, which correspond to equilibria of the reduced system.  Thus, $\boldsymbol m, q,p$ are constant, and hence so is $\boldsymbol\omega$. Consequently, $\theta$ and $\varphi$ are also constant, and only $\psi$ varies, and it varies uniformly. This shows (as expected) that the relative equilibrium consists of a uniform rotation about the axis containing $\boldsymbol M$.

\newpage
\section{Reduction for the Lobachevski plane $L^2$}
\label{S:Red L2}

We now consider the setting of two particles on the Lobachevsky plane $L^2$ (also called the pseudo-sphere, or the hyperbolic plane).  The reduction procedure is very similar to the spherical setting above, so we give fewer details.  

Consider the three-dimensional Minkowski space and attach the fixed coordinate system $OXYZ$ to it. The scalar product
is given by
\begin{equation}
\label{eq1.7}
\langle\boldsymbol R, \boldsymbol R\rangle_\mathrm{K}=(\boldsymbol R, {\mathrm{K}}\boldsymbol R), \quad {\mathrm{K}}=\text{diag}(1, 1, -1),
\end{equation}
where again $(\cdot ,\cdot)$ is the usual Euclidean scalar product.
The Lobachevsky plane is defined  as the  Riemannian manifold that consists of the upper sheet of the hyperboloid
\begin{equation}
\label{eq1.8}
\langle\boldsymbol R, \boldsymbol R\rangle_\mathrm{K}=X^2+Y^2- Z^2=-1,
\end{equation}
equipped with the metric that it inherits from the ambient Minkowski space.
The Gaussian curvature of $L^2$ equals $-1$. A standard expression for the metric is 
$$
ds^2=d\theta^2+\sinh^2\theta\, d\varphi^2,
$$
and may be obtained by  choosing  local coordinates on~\eqref{eq1.8} in the form
$$
X=  \sinh \theta \cos \varphi, \quad Y=  \sinh \theta \sin \varphi, \quad Z= \cosh \theta,
$$
and then restricting~\eqref{eq1.7}. The group of orientation preserving  isometries of $L^2$ is $SO(2,1)$ (the set of $3\times 3$ invertible matrices $g$
satisfying $g^T\mathrm{K}g=\mathrm{K}$) acting naturally on $(X,Y,Z)^T\in L^2$.

 The setup of the  $2$-body problem on $L^2$ is analogous to the one described for $S^2$. In this case,
 the generalization of the inverse square law is given by the potential
\begin{equation}
\label{eq:grav-potential-L2}
U_{\mathrm{grav}}(q)=-\frac{G\mu_1\mu_2}{\tanh q},
\end{equation}
where $q\in (0,\infty)$ is the Riemannian distance between the particles. As for the spherical case, this form of the potential leads to Bertrand's property in Kepler's problem
 and to a natural generalization
of Kepler's first law  \cite{kozlov, Carinena}. We consider more general attractive potentials $U=U(q)$ satisfying $U'(q)>0$ for all $q\in (0,\infty)$ and $U(q)\to -\infty$ as $q\to 0^+$.

The configuration space of the problem is 
\begin{equation}
\label{eq:Q for L2}
Q=L^2\times L^2  \setminus \Delta
\end{equation}
 where $\Delta$ is the collision set. 
The momentum phase space is the eight-dimensional manifold $T^*Q$, with Hamiltonian $H=T+U$. The kinetic energy $T$
is the sum of the kinetic energies of the particles. Each of them is obtained as the product of the mass with the norm squared of the 
velocity  vector, where the norm is computed with respect to the Riemannian metric. 

The problem is invariant under the action of $SO(2,1)$ which simultaneously ``rotates" the particles. In analogy with Theorem~\ref{th:reduction S2} we have:
\begin{theorem}
\label{th:reduction L2}
For $Q$ given by \eqref{eq:Q for L2},
the reduced space $T^*Q/SO(2,1)$ is isomorphic as a Poisson manifold to $\R^3\times (0,\infty)\times \R\ni ({\boldsymbol m},q,p)$. The Poisson structure on this space is defined by the relations,
\begin{equation}
\label{eq:PB on L2}
\{m_x, m_y\}=m_z, \quad \{m_y, m_z\}=-m_x, \quad 
\{m_z, m_x\}=-m_y, \quad \{q, p\}=1,
\end{equation}
where $\boldsymbol m=(m_x,m_y,m_z)^T$.

For the 2-body problem on the Lobachevsky plane with interaction governed by potential energy $U(q)$, the reduced dynamics is Hamiltonian with respect to this Poisson structure  
with Hamiltonian 
\begin{equation}
\label{eq:Hamiltonian on L2}
H({\boldsymbol m},q,p)=\frac{1}{2\mu_1 }\left ( (\boldsymbol  m , {\bf A}(q) \boldsymbol  m)
-2 m_x p + ( 1 +  \mu)p^2 \right ) + U(q),  
\end{equation}
where $(\cdot, \cdot)$ is the Euclidean  scalar product  and
\begin{equation}
\label{eq:defAHyper}
 {\bf A}(q) =\begin{pmatrix}  1 & 0 & 0 \\ 0 & 1 &  - \dfrac{\cosh q}{\sinh q} \\
0 & - \dfrac{  \cosh q}{\sinh q}  & \dfrac{ (\mu +\cosh^2 q)}{ \sinh^2 q} \end{pmatrix}, \qquad \mu=\frac{\mu_1}{\mu_2}.
\end{equation} 
\end{theorem}
As a consequence of the theorem, the reduced equations of motion are
\begin{equation}
\label{eq2.4}
\dot{\boldsymbol  m}=(\mathrm{K}\boldsymbol  m)\times \frac{\partial H}{\partial \boldsymbol  m}, \quad \dot{q}=\frac{\partial H}{\partial p}, \quad
\dot{p}=-\frac{\partial H}{\partial q}.
\end{equation}

 The Poisson bracket \eqref{eq:PB on L2} has generic rank 4 (its rank drops to 2 when $\boldsymbol m=0$). Its 4-dimensional symplectic leaves are the
regular level sets of the Casimir function
\begin{equation} \label{eq:Casimir for L^2}
C(\boldsymbol  m) =-\langle \boldsymbol  m, \boldsymbol  m \rangle_\mathrm{K}= -m_x^2 -m_y^2+m_z^2,
\end{equation}
which is a first integral of the reduced equations of motion  \eqref{eq2.4}.

The  following 2 subsections give a proof of Theorem~\ref{th:reduction L2}.

\subsection{Group parametrization of configurations}
\label{SS:L2 conventions}

The proof of Theorem~\ref{th:reduction L2} proceeds as in the spherical case. We introduce a moving  frame $Oxyz$ 
 which is orthogonal in the metric~\eqref{eq1.7}, such that the point masses have body coordinates 
$$
 \boldsymbol r_{1}=( 0, 0, 1)^T, \qquad \boldsymbol r_{2}=( 0, \sinh q, \cosh q)^T,
$$
 where   $q\in (0, \infty)$ is the hyperbolic distance between the masses. As for the spherical case,
 the moving frame is completely determined by the requirement that $\mu_1$ is on the $z$-axis and $\mu_2$ lies on the
 positive $Oyz$ plane. Hence, as  a manifold,  $Q=SO(2,1)\times (0,\infty)$.

We introduce local coordinates for the group $SO(2, 1)$ by adapting the  Euler angles for $SO(3)$ used 
in Section~\ref{SS:Group-parametrization} (see Fig.\,\ref{figris}).  Consider an element $g\in SO(2,1)$ that 
relates the moving and the fixed frame  as a sequence of 3 rotations $g=R^Z_\psi R^X_\theta R^Z_\varphi$, with the second rotation $R^X_\theta$ being  hyperbolic. This leads to 
\begin{equation}
\label{eq:Euler for L2}
\begin{small}
g(\varphi, \theta, \psi)=
\begin{pmatrix}
\cos \varphi \cos \psi-\cosh \theta \sin\varphi \sin \psi & -\sin \varphi \cos \psi-\cosh \theta \cos\varphi \sin \psi & -\sinh \theta \sin \psi \\
\cos \varphi \sin \psi+\cosh \theta \sin \varphi \cos \psi & -\sin \varphi \sin \psi+\cosh \theta \cos \varphi \cos \psi & \sinh \theta \cos \psi \\
\sinh \theta \sin \varphi & \sinh \theta \cos \varphi & \cosh \theta
\end{pmatrix},
\end{small}
\end{equation}
with $0\leq\varphi, \psi <2\pi$, $\theta \in [0,\infty)$. In this way, $(\theta, \varphi, \psi,q)$ are generalized coordinates for the problem. The 
position of the masses with respect to the fixed coordinate system are given by $\boldsymbol R_{1}= g \boldsymbol r_{1}$ and $\boldsymbol R_{2}= g \boldsymbol r_{2}$,
and the kinetic energy of the system is
\begin{equation}
\label{eq:KinEnergyL2}
T=\frac{1}{2} \left (\mu_1 \langle  \dot{ \boldsymbol R}_{1} , \dot{ \boldsymbol R}_{1}  \rangle_\mathrm{K}  + \mu_2  \langle  \dot{ \boldsymbol R}_{2} , \dot{ \boldsymbol R}_{2}  \rangle_\mathrm{K} \right ).
\end{equation}

\begin{figure}[h]
\centering
\includegraphics[totalheight=5cm]{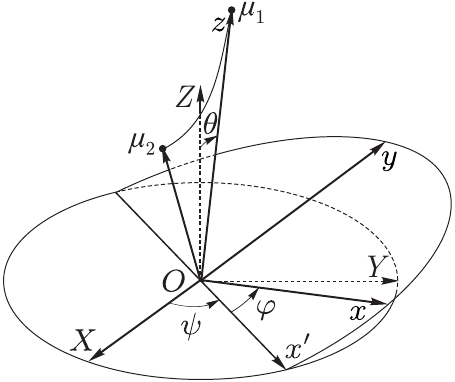}
\caption{Generalized Euler angles for the 2-body configuration on $L^2$.}
\label{figris}
\end{figure}

Analogous to \eqref{eq:ang velocity}, the  left-invariant angular velocity $\boldsymbol  \omega\in \R^3$ is defined by $\widetilde{\boldsymbol  \omega}=g^{-1}\dot g$. Here $\widetilde{\boldsymbol  \omega}:=\widehat{\boldsymbol \omega} \mathrm{K}$ where,
 as usual,  $\widehat{\boldsymbol \omega}$ is
given by \eqref{eq:hatmap}. Performing the algebra we find
\begin{equation}
\label{eq:ang vel L2}
\omega_x=\dot \theta \cos\varphi+\dot\psi \sinh\theta \sin\varphi, \quad \omega_y=-\dot \theta \sin\varphi+\dot\psi \sinh\theta \cos\varphi,
\quad \omega_z=\dot \varphi+\dot \psi\cosh \theta.
\end{equation}
The kinetic energy  $T$ is written in terms of $(\boldsymbol \omega, \dot q, q)$ as
\begin{equation*}
T= \frac{\mu_1}{2} ({\bf A}(q)^{-1}  \boldsymbol \omega,  \boldsymbol \omega)   +\frac{\mu_2}{2} (\omega_x + \dot q)^2 ,
\end{equation*}
with ${\bf A}(q)$  given by \eqref{eq:defAHyper}.

\subsection{Generalized momenta and reduction}
\label{SS:reductionL2}

The procedure now is exactly as in Section~\ref{SS:reductionS2}. After the introduction of the
angular momentum vector in the moving frame $\boldsymbol m :=\frac{\partial T}{\partial  \boldsymbol \omega}$ and the canonical
momenta $(P_\theta, P_\varphi, P_\psi,p)$ that are conjugate to $(\theta, \varphi, \psi, q)$, the relevant formulae to complete the proof
of Theorem~\ref{th:reduction L2} are
\begin{equation}
\label{eq:m in terms of p L2}
\boldsymbol m =\frac{1}{\sinh \theta} \begin{pmatrix} 
 \sin\varphi\,(P_\psi-P_\varphi\cosh\theta)+P_\theta\sinh \theta \cos\varphi \\
\cos\varphi\,(P_\psi-P_\varphi\cosh\theta)-P_\theta\sinh \theta \sin\varphi \\
 \sinh \theta P_\varphi ,
\end{pmatrix},
\end{equation}
and 
\begin{equation*}
\label{eq3.01}
T=\frac{1}{2\mu_1 }\left ( (\boldsymbol  m , {\bf A}(q) \boldsymbol  m)
-2 m_x p + ( 1 +  \mu)p^2 \right ).
\end{equation*}

\begin{remark}The geometric interpretation is also analogous to the spherical case. 
We have worked out the left trivialization of  $T^*SO(2,1)$ to arrive at the decomposition
\begin{equation*}
T^*Q= T^*SO(2,1) \times T^*(0,\infty) = SO(2,1) \times so(2,1)^*  \times (0,\infty) \times \R.
\end{equation*}
We  then eliminated the $SO(2,1)$ component to obtain a bracket on $T^*Q/SO(2,1)$ that is the product of the Lie-Poisson bracket on $so(2,1)^*$ and the
canonical bracket on $T^*(0,\infty)=(0,\infty) \times  \R$.
\end{remark}

\subsection{Conserved quantities and reconstruction}
\label{SS:reconstruction-L2}  

\def\m{{\boldsymbol m}}
\def\M{{\boldsymbol M}}

As for the spherical case, we indicate how to determine the time evolution of $(\theta, \varphi, \psi)$ assuming
that a solution $(\boldsymbol m(t), q(t), p(t))$ to the reduced system~\eqref{eq2.4} is known.

By the $SO(2,1)$  invariance of the problem,  the angular momentum $\boldsymbol M$
of the system written in the  fixed axes is constant along the motion. According to our previous definitions we may write
\[
\boldsymbol M:=\sum\limits^{2}_{\alpha=1}\mu_\alpha(  \mathrm{K} \boldsymbol R_\alpha) \times ( \mathrm{K} \dot{\boldsymbol R}_\alpha),
\]
and we have $\boldsymbol m(t)  = g(t)^{-1}\boldsymbol M$.

In contrast to the spherical setting, here the geometry is more subtle and for example we cannot simply choose $\boldsymbol M = (0,0,M_0)^T$ as this is insufficiently general.  For a given $\m(0)$ the momentum vector $\m(t)$ evolves on a surface with different possible geometries depending on the sign of the Casimir function $C$ given 
by~\eqref{eq:Casimir for L^2} . See Fig.\,\ref{fig:hyperboloids}. If $\m\neq0$ there are three types:
\begin{itemize}
\item \emph{Elliptic momentum}  Here $C(\m)>0$, and the momentum evolves on one sheet of a two-sheeted hyperboloid (the sheets are distinguished by the sign of $m_z$).  For a motion with elliptic momentum, one can choose a frame so that $\M=(0,0,M_0)^T$, with $M_0\in\R$.  Since such an $\M$ defines a point on the hyperbolic plane $L^2$, the subgroup of transformations in $SO(2,1)$ that preserve an elliptic momentum fixes that point in $L^2$ and hence is an elliptic subgroup; it is the group of rotations about $\M$ and is isomorphic to $SO(2)$. 
\item \emph{Hyperbolic momentum}  Here $C(\m)<0$, and the momentum evolves on  a one-sheeted hyperboloid. Here one can choose $\M=(0,M_0,0)^T$ with $M_0>0$.  The subgroup of transformations in $SO(2,1)$ that preserve a hyperbolic momentum is a hyperbolic subgroup, and is the group of hyperbolic rotations about $\M$ which is isomorphic to $\R$. 
\item \emph{Parabolic momentum} Here $C(\m)=0$, and the momentum evolves on one sheet of the cone with the origin removed (again, the two sheets are distinguished by the sign of $m_z$).  Here one can chose $\M=\pm(1,0,1)^T$, and the subgroup preserving such a vector is a parabolic subgroup and is also isomorphic to $\R$. 
\end{itemize}

Note that in Minkowski geometry, elliptic, hyperbolic and parabolic vectors are often called timelike, spacelike and lightlike vectors, repectively.

\begin{figure}
\centering
\includegraphics[width=3cm]{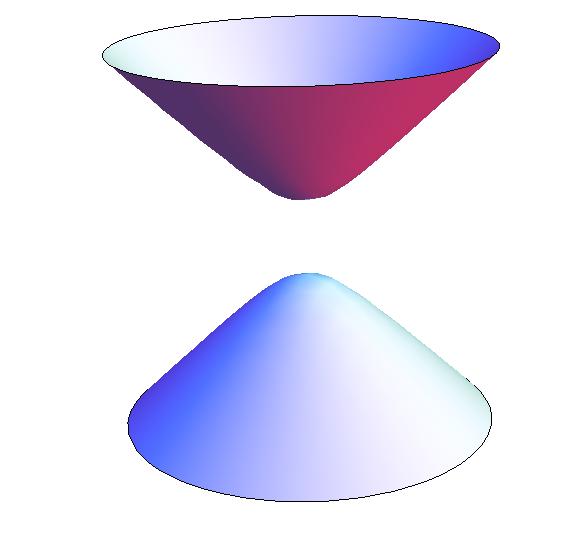}\includegraphics[width=3cm]{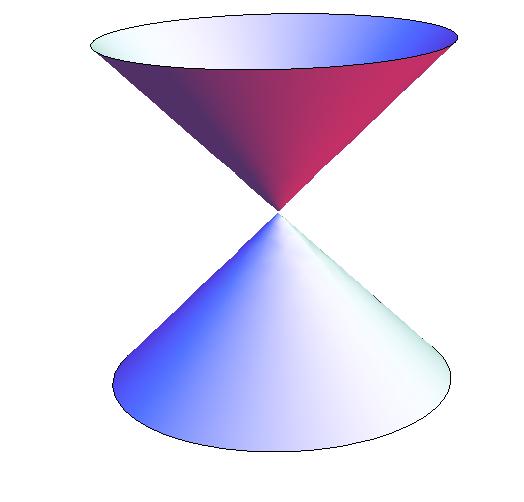}\includegraphics[width=3cm]{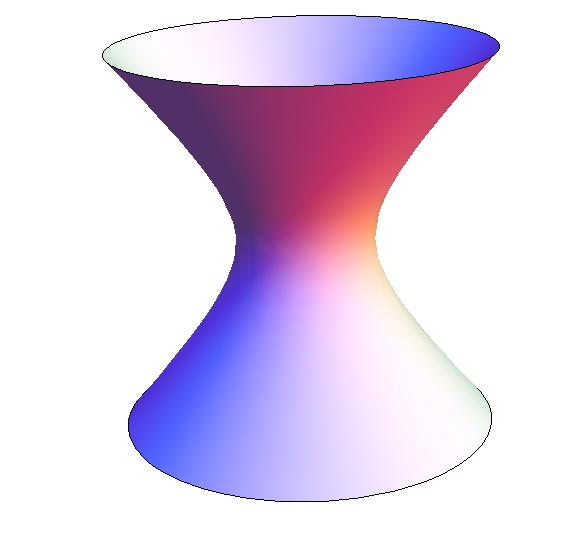}
\begin{minipage}{0.8\textwidth}
\caption{The three possible types of non-zero momentum, from left to right: elliptic, parabolic and hyperbolic, determined by the sign of the Casimir \eqref{eq:Casimir for L^2}.}
\label{fig:hyperboloids}
\end{minipage}
\end{figure}

\medskip

\noindent\emph{Elliptic momentum.} Consider  a solution $(\boldsymbol{ m}(t),q(t),p(t))$ of the reduced system \eqref{eq2.4} having $m_z(t)>0$.
In this case, by a choice of inertial frame, we can assume $\boldsymbol M= (0,0,M_0)^T$ with $M_0> 0$, and using \eqref{eq:Euler for L2}
we get
\begin{equation} \label{eq:m for L2}
\boldsymbol m(t)  = g(t)^{-1}\boldsymbol M = M_0 \begin{pmatrix} -\sinh\theta \sin\varphi \\ - \sinh\theta\cos\varphi \\ \cosh\theta  \end{pmatrix} .
\end{equation}
Therefore,
\begin{equation}
\label{eq:reconstruction-sphere-thetaphi}
\cosh\theta=\frac{m_{z}(t)}{M_0},\quad
\tan\varphi=\frac{m_{x}(t)}{m_{y}(t)}.
\end{equation}
The quadrature for the evolution of~$\psi$ is in this case obtained by using~\eqref{eq:ang vel L2} to write
\reversemarginpar
\begin{equation}
\begin{gathered}
\label{eq1.9ehyp}
\dot{\psi}=\frac{\sin\varphi \, \omega_x+\cos\varphi\, \omega_y}{\sinh\theta}=
-\frac{M_0\left(m_{x}(t)\omega_{x}(t)+
m_{y}(t)\omega_{y}(t)\right)}{M^{2}_{0}-m_{z}^2(t)},\\
\end{gathered}
\end{equation}
where 
\begin{equation*}
\omega_x=\frac{\partial H}{\partial m_x},\quad
\omega_y=\frac{\partial H}{\partial m_y}.
\end{equation*}

At a relative equilibrium of elliptic type, since $\boldsymbol m$, $q,p$ are constant, we deduce that $\varphi,\theta$ are constant, while $\psi$ varies uniformly. Thus an elliptic relative equilibrium consists of a uniform rotation about the $\boldsymbol M$-axis, and hence is a periodic orbit.

\bigskip

\noindent\emph{Hyperbolic momentum} 
In this case, by a choice of inertial frame, we can assume $\boldsymbol M= (M_0,0,0)^T$ with $M_0> 0$.  For the reconstruction it is convenient to use an alternative set of hyperbolic Euler angles defined by writing $g=R^X_\kappa R^Z_\psi R^X_\theta$, where the
rotations with respect to $\kappa$ and $\theta$ are hyperbolic.  This leads to:
\begin{equation}
\label{eq:Euler Angle L2 New}
g=\begin{pmatrix}
  \cos \psi &-\sin\psi \cosh\theta &-\sin\psi \sinh\theta \\ 
  \cosh \kappa  \sin \psi &
  \cosh \kappa \cos \psi \cosh \theta +\sinh \kappa \sinh \theta &
  \cosh \kappa \cos\psi \sinh \theta +\sinh \kappa \cosh \theta \\ 
  \sinh\kappa \sin \psi &
  \sinh \kappa\cos \psi \cosh \theta +\cosh \kappa \sinh \theta &
  \sinh \kappa \cos \psi \sinh \theta +\cosh \kappa \cosh \theta 
\end{pmatrix},
\end{equation}
in place of \eqref{eq:Euler for L2}  (note that, unlike (\ref{eq:Euler for L2}), this representation is not global). Putting $\m(t)=g^{-1}(t)\M$ with $\M=(M_0,0,0)$ one finds
\begin{equation}
 \label{eq:m for L2 New}
\m(t)=g^{-1}(t)\M = M_0\begin{pmatrix} \cos\psi \\ -\sin\psi \cosh\theta \\ \sin\psi
\sinh\theta\end{pmatrix}.
\end{equation}
This choice of Euler angles only allows reconstruction for $|m_x(t)|\leq M_0$, but that will suffice for reconstructing the relative equilibria we find for this system. Therefore,
\begin{equation*}
\cos\psi =\frac{m_{x}(t)}{M_0},\quad
\tanh\theta=-\frac{m_{z}(t)}{m_{y}(t)}.
\end{equation*}
The quadrature for the evolution of~$\kappa$ is in this case obtained by noting that the angular velocity $\boldsymbol \omega$ (following the same construction leading to \eqref{eq:ang vel L2}) is given by
\begin{equation*}
\boldsymbol{\omega}=\left (\dot \theta +\dot \kappa \cos \psi , \, \dot \psi \sinh \theta -\dot \kappa  \cosh \theta \sin \psi,\, \dot \psi \cosh \theta - \dot \kappa \sinh \theta \sin \psi \right ),
\end{equation*}
and therefore,
\reversemarginpar
\begin{equation}
\label{eq:dkappa}
\begin{gathered}
\dot{\kappa}=\frac{- \cosh \theta  \, \omega_y+\sinh \theta \, \omega_z}{\sin \psi}=
\frac{M_0\left(m_{y}(t)\omega_{y}(t)+
m_{z}(t)\omega_{z}(t)\right)}{M^{2}_{0}-m_{x}^2(t)},\\
\end{gathered}
\end{equation}
where 
\begin{equation*}
\omega_y=\frac{\partial H}{\partial m_x},\quad
\omega_z=\frac{\partial H}{\partial m_y}.
\end{equation*}

\newpage

\section{Relative equilibria for the  2-body problem on the Lobachevsky plane}

Consider two masses $\mu_1$ and $\mu_2$ on the Lobachevsky plane $L^2$, interacting via an attracting conservative force with potential energy $U(q)$
as introduced in Section~\ref{S:Red L2} where $q$ is the Riemannian distance between the masses and where by attracting we mean $U'(q)>0$ for all $q\in(0,\infty)$.  
Recall that the mass ratio is denoted by $\mu=\mu_1/\mu_2$.

\subsection{Classification of relative equilibria}
\label{SS:Classification L2}

The symmetry group of this problem is $SO(2,1)$ (or equivalently $SL(2,\R)$) and has three types of non-trivial 1-parameter subgroup 
(as is well known from hyperbolic geometry), and correspondingly 3 types of relative equilibrium (see Definition\,\ref{def:RE}):
\begin{itemize}
\item \emph{elliptic subgroup}: here the subgroup is compact and isomorphic to $SO(2)$, and hence the motion is periodic; REs of this type can only occur if $C>0$.
\item \emph{hyperbolic subgroup}: here the subgroup is not compact, consists of semisimple matrices, and is isomorphic simply to $\R$ and hence the motion is unbounded; REs of this type can only occur if $C<0$.
\item \emph{parabolic subgroup}: this is also non-compact, so unbounded trajectories, but unlike the first two, the elements are not semisimple; and an RE of parabolic type would require $C=0$ (although we show below there are none for the 2-body system).
\end{itemize}
The relation between the type of subgroup and the sign of the Casimir $C$ is described in Section\,\ref{SS:reconstruction-L2}.

\begin{theorem}\label{thm:existence on L2}
In the 2-body problem on the Lobachevsky plane, governed by an attractive force, for each value of $q>0$ there are precisely two relative equilibria where the particles are a distance $q$ apart, one of elliptic and one of hyperbolic type.  
\end{theorem}

Elliptic and hyperbolic RE are illustrated as bifurcation curves on the energy-momentum diagram in Fig.\,\ref{F:DiagramL2}.
Elliptic RE correspond the  curve with a cusp on the half-plane $C>0$  whereas hyperbolic RE 
correspond to the smooth bifurcation curve on the half-plane $C<0$. These curves meet each other
at the punctured point when $C=0$ which corresponds to the non-existence of parabolic RE.

The qualitative properties of the motion along these RE is explained in subsection~\ref{ss:reconstructionRE-L2}. The results of 
Theorem~\ref{thm:existence on L2} were given before in   \cite{DiacuPCh,LGN2016} for the gravitational potential \eqref{eq:grav-potential-L2}. Here
we extend the classification to arbitrary attractive potentials.

\subsection*{Proof of Theorem\,\ref{thm:existence on L2}}

 Relative equilibria correspond to equilibrium points of the reduced system~\eqref{eq2.4}, so they are  solutions of the following set of equations:
\begin{subequations}\label{eq:REH}
\begin{align}
\label{eq:REH-p}
\frac{\partial H}{\partial p }&=0, \\
\label{eq:REH-m}
 (  {\mathrm{K}} \boldsymbol   m) \times \frac{\partial H}{\partial \boldsymbol  m}& ={\bf 0}, \\
\label{eq:REH-q}
\frac{\partial H}{\partial q} & =0,
\end{align}
\end{subequations}
where the Hamiltonian $H$ is given by~\eqref{eq:Hamiltonian on L2}, the matrix ${\mathrm{K}}$ is defined in~\eqref{eq1.7} and $\times$ is the vector product in $\R^3$.
The condition \eqref{eq:REH-p} yields
\begin{equation}
\label{eq:pH}
p=\frac{m_x}{1+\mu}.
\end{equation}
Substituting this expression into the last two components of \eqref{eq:REH-m} yields two possibilities:
\begin{enumerate}
\item[{1)}] $m_y=m_z=0$. In this case \eqref{eq:REH-q} implies  $U'(q)=0$, and  there is no solution for $q$ by our assumption that $U$ is attractive.
\item[{2)}] $m_x=0$. We analyse  this case in what follows assuming that $m_y$ and $m_z$ do not vanish simultaneously, since otherwise we are back in case (i).
\end{enumerate}

\emph{Elliptic relative equilibria.} Recall that we are considering the case $m_x=0$. We parametrise the open region of the reduced phase space having
$C(\boldsymbol m)>0$
with the
parameters $M \neq 0$ and $\alpha \in \R$, by putting
\begin{equation}
\label{eq:mymzEH}
m_y=M \sinh \alpha, \qquad m_z=M\cosh \alpha.
\end{equation}
Then $C(\boldsymbol m)=M^2$ and \eqref{eq:REH-m} is satisfied provided that
\begin{equation}
\label{eq:CofMH}
\sinh 2(q-\alpha)=\mu \sinh 2\alpha.
\end{equation}
Equation \eqref{eq:CofMH} admits the unique solution for $\alpha$
\begin{equation}
\label{eq:alphaq}
\alpha = \frac{q}{2} +\frac{1}{4}  \ln \left ( \frac{\mu+e^{2q}}{1+\mu e^{2q}}\right ).
\end{equation}
With the above value of $\alpha$, equation \eqref{eq:REH-q} is satisfied provided that $M$ is such that
\begin{equation}
\label{eq:M0EH}
M^2=\frac{\mu_1 \sinh^3 q \,U'(q)}{\mu \cosh^2\alpha \cosh q + \cosh\alpha \cosh (q- \alpha)}.
\end{equation}
For an attractive potential, the right hand side of this expression is strictly positive. The choice of sign for $M$ corresponds to two solutions related by  the  time reversing
symmetry of the system, and we do not distinguish them in our classification. 

\emph{Hyperbolic relative equilibria.} The analysis is analogous to the above. This time we put
\begin{equation}
\label{eq:mymzHH}
m_y=M\cosh \alpha, \qquad m_z=M\sinh \alpha,
\end{equation}
so $C(\boldsymbol  m)=-M^2<0$. Taking into account that $m_x=0$, then  \eqref{eq:REH-m} is satisfied provided that \eqref{eq:CofMH},  and hence also  \eqref{eq:alphaq}, hold.
The condition for $M$ in this case is
\begin{equation}
\label{eq:MoHypH}
M^2=\frac{\mu_1 \sinh^3q U'(q)}{\mu \sinh^2\alpha \cosh q - \sinh\alpha \sinh (q- \alpha )}.
\end{equation}
Using \eqref{eq:CofMH} to eliminate $\mu$, one can write the denominator of the above expression as
\begin{equation*}
\begin{split}
 \frac{\sinh^2\alpha \sinh(q-\alpha)}{\sinh 2\alpha} \left ( 2 \sinh^2(q-\alpha)\cosh \alpha + \sinh \alpha \sinh(2(q-\alpha)) \right ).
\end{split}
\end{equation*}
Considering that \eqref{eq:alphaq}  implies $0<\alpha<q$, it is immediate to check that all of the terms on the right hand side of this expression are positive. Therefore, the right hand side of~\eqref{eq:MoHypH} is also positive and 
 there is a unique solution for $M$ (modulo the time reversibility symmetry of the problem).

\emph{Parabolic relative equilibria.} Finally we show that there are no solutions of \eqref{eq:REH-m} having $C(\boldsymbol  m)=0$. Substituting $m_x=0$ and  $m_y= \pm  m_z$
into the first component of \eqref{eq:REH-m} yields, after a simple calculation,
\begin{equation*}
\cosh 2q\pm \sinh 2q+\mu=0,
\end{equation*}
which clearly has no solutions for $q$ since $\mu>0$.

\subsection{Reconstruction} \label{ss:reconstructionRE-L2}

Here we relate the (relative) equilibria described above in the reduced space to the corresponding motion in the original unreduced space.  
The qualitative properties of these unreduced RE depend on the sign that $C$ takes along them.
\begin{enumerate}
\item Fix $q>0$ and consider the corresponding elliptic RE. From the discussion above, the angular momentum $\boldsymbol m$ satisfies $m_x=0$ and 
$m_y$ and $m_z$ are given by \eqref{eq:mymzEH} where $\alpha>0$ is determined by \eqref{eq:alphaq} and $M$ by \eqref{eq:M0EH}.
We use the generalised Euler angle convention given by \eqref{eq:Euler for L2} for the reconstruction.
The equations  \eqref{eq:m for L2}  are satisfied for the constant value of  the angles $\varphi=\pi$ and $\theta =\alpha$ (the other possibility,
namely that $\varphi=0$ and $\theta =-\alpha$, leads to a solution in the same group orbit of the RE - see the discussion after Definition~\ref{def:RE}). 
The  angle $\psi =\omega t$ with constant $\omega:=\dot \psi$  given by \eqref{eq1.9ehyp}. Using \eqref{eq:M0EH} and \eqref{eq:CofMH} it is possible to simplify 
\begin{equation*}
\omega^2 = \zeta^{-1}U'(q),
\end{equation*}
 where $\zeta=\frac{\mu_1\sinh 2\alpha}{2}=\frac{\mu_2\sinh 2(q-\alpha)}{2}$. According to the conventions of section~\ref{SS:L2 conventions}, the
  position of the particles in the fixed frame along the RE is
  \begin{equation*}
{\boldsymbol R}_1(t)= \begin{pmatrix}  -\sinh \alpha \sin \omega t \\ \sinh \alpha \cos \omega t \\ \cosh \alpha \end{pmatrix}, \qquad
 {\boldsymbol R}_2(t)=\begin{pmatrix}  \sinh (q-\alpha) \sin \omega t \\ - \sinh (q-\alpha) \cos \omega t \\ \cosh (q-\alpha) \end{pmatrix}.
\end{equation*}
The motion is then periodic with constant angular speed 
 $\omega$. We note that throughout the motion, the particle $\mu_1$ (respectively $\mu_2$) has constant distance $\alpha$ (respectively $q-\alpha$) to the  
point with space coordinates $(0,0,1)$. This point corresponds to the {\em centre of mass} as defined in \cite{LGN2016}.
 Using \eqref{eq:alphaq} one may check that this point is closer to the heavier mass. Hence, elliptic RE in $L^2$ generalise the RE 
 in flat space where the centrifugal forces are balanced by gravitational attraction.
 
 In the next section we show that these motions are stable if $0<q<q^*$ and unstable if $q>q^*$. The critical distance $q^*$ depends on the mass ratio $\mu$ as 
 indicated in the statement of Theorem~\ref{S:Stability-hyperbolic} below.

\item Now consider the hyperbolic RE corresponding to $q>0$.  The angular momentum $\boldsymbol m$ satisfies $m_x=0$ and 
$m_y$ and $m_z$ are given by \eqref{eq:mymzHH} where $\alpha>0$ is determined by \eqref{eq:alphaq} and $M$ by \eqref{eq:MoHypH}.
This time we use the generalised Euler angle convention given by \eqref{eq:Euler Angle L2 New}
for the reconstruction. Equation \eqref{eq:m for L2 New} implies that  $\psi$ and $\theta$ have the constant
values $\theta=-\alpha$, $\psi=3\pi/2$ along the motion. On the other hand $\kappa=\omega t$ where $\omega:=\dot \kappa$ is
defined by  \eqref{eq:dkappa}. Using 
\eqref{eq:MoHypH} and \eqref{eq:CofMH}, one may simplify
\begin{equation*}
\omega^2 = \zeta^{-1}U'(q),
\end{equation*}
where, as before, $\zeta=\frac{\mu_1\sinh 2\alpha}{2}=\frac{\mu_2\sinh 2(q-\alpha)}{2}$. 
This time, the  position of the particles in the fixed frame along the RE is
  \begin{equation*}
{\boldsymbol R}_1(t)= \begin{pmatrix}  \sinh \alpha  \\ \cosh \alpha \sinh \omega t \\ \cosh \alpha \cosh \omega t \end{pmatrix}, \qquad
 {\boldsymbol R}_2(t)=\begin{pmatrix} \sinh (q+\alpha)  \\ \cosh (q+\alpha) \sinh \omega t \\ \cosh (q+\alpha) \cosh \omega t \end{pmatrix}.
\end{equation*}

Note that the motion along these RE is unbounded.  This type of relative equilibrium does not exist in the positive or zero curvature case. As explained in \cite{LGN2016} their existence is due to the property of the Lovachevsky space that makes parallel geodesics ``separate''. This separating effect is balanced by the attractive forces in a very delicate manner.   We show below that these RE are all unstable. 
\end{enumerate}

\subsection{Stability analysis of the relative equilibria on the Lobachevsky plane}
\label{S:Stability-hyperbolic}
The stability of the relative equilibria depends on the form of the potential $U(q)$. We restrict attention here to the gravitational potential given in \eqref{eq:grav-potential-L2}.
We assume that $\mu_1=1$ and $G\mu_1\mu_2=1$,\footnote{This assumption is done without loss of generality
since one may eliminate these quantities from the equations of motion \eqref{eq2.4} by  rescaling time $t\to \frac{\sqrt{\varkappa} t}{\mu_1}$,
and the momenta  $p\to \frac{p}{\sqrt{\varkappa}} , \; {\boldsymbol m}\to \frac{ \boldsymbol m}{\sqrt{\varkappa}}$, where $\kappa=\mu_1\sqrt{G\mu_2}$.} so
the only external parameter left in the problem is the mass ratio $\mu$.

\begin{theorem} \label{thm:stability on L^2}
Of the relative equilibria described in Theorem\,\ref{thm:existence on L2} above, and assuming a potential of the form \eqref{eq:grav-potential-L2}:
 \begin{enumerate}
\item  All hyperbolic RE are  (linearly) unstable. 
\item  There is a critical distance $q^*>0$ depending on the mass ratio $\mu$ such that when the distance between the particles is less than $q^*$ the elliptic 
RE is nonlinearly  stable, and when they are more than $q^*$ apart it is unstable. \newline
The critical value $q^*$ is given by 
\begin{equation*}
q^*=\alpha^* + \frac{1}{2}\mathrm{ arcsinh} (\mu( \sinh 2\alpha^*))
\end{equation*}
where $\alpha^*$ is the unique positive solution of the equation
\begin{equation*}
2\sinh^2 \alpha +1= 2\sinh^2 \alpha \sqrt{1+\mu^2\sinh^2 2\alpha}.
\end{equation*}
\end{enumerate}
\end{theorem}

We prove this theorem\footnote{An equivalent proof of Theorem\,\ref{thm:stability on L^2} was  given before in~\cite{LGN2016} by working on a symplectic slice in the unreduced system, since the reduced equations of motion were not known at that time.  Although equivalent, the approach that we follow in this paper is elementary.}
 by analyzing the signature of the Hessian matrix of the Hamiltonian function restricted to to the symplectic leaves.
According to Lyapunov's Theorem, if a RE is a local maximum or minimum of  the Hamiltonian, then it is nonlinearly stable.
This happens in particular if the corresponding Hessian matrix  at the equilibrium is positive or negative definite.
On the
other hand, if the Hessian matrix  has an odd number of negative eigenvalues, then the linearised system
has at least one eigenvalue with positive real part and  the corresponding RE
is unstable. The proof  of Theorem\,\ref{thm:stability on L^2} follows at once from these observations and the following proposition.
\begin{proposition}
\label{P:Signature-L2}
Let  $q^*$  defined  as in the statement of Theorem~\ref{thm:stability on L^2}. The Hessian matrix of the reduced Hamiltonian  (restricted to the corresponding symplectic leaf)
at the RE of the problem described in Theorem\,\ref{thm:existence on L2}, has the following signature:
\begin{enumerate}
\item All hyperbolic RE  have signature $(+++-)$.
\item Elliptic RE with $0<q<q^*$ have signature $(++++)$.
\item Elliptic RE with $q^*<q $ have signature $(+++-)$.
\end{enumerate}
Moreover, the Casimir function $C$ restricted to the branch of elliptic RE has a maximum at $q=q^*$.
\end{proposition}

The results described above are illustrated in the energy-momentum diagram Fig.\,\ref{F:DiagramL2} where we have indicated
the signature of the Hessian along the branches of RE. The critical value $q^*$ at which the elliptic RE undergo a saddle-node bifurcation corresponds to the cusp where
$C$ has a maximum. Fig.\,\ref{F:Bif-diag-ellipticRE} shows the stability region on the $q$-$\mu$ plane for elliptic RE. Such region is delimited by the curve $q=q^*(\mu)$ .
\begin{figure}[h]
\centering
\includegraphics[width=7cm]{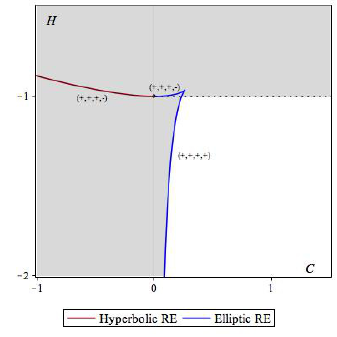}
\begin{minipage}{0.8\textwidth}
\caption{Energy-Momentum bifurcation diagram of relative equilibria on $L^2$ for the gravitational potential
with $\mu=1/2$. The shaded area on the $C$-$H$ plane shows all possible values of $(C,H)$. We also indicate the signature of the Hessian matrix of the Hamiltonian along each branch of relative equilibria.
Notice the change in signature at the cusp of the elliptic relative equilibria where a saddle-node bifurcation takes place. }
\label{F:DiagramL2}
\end{minipage}
\end{figure}

\begin{figure}[h]
\centering
\includegraphics[width=5cm]{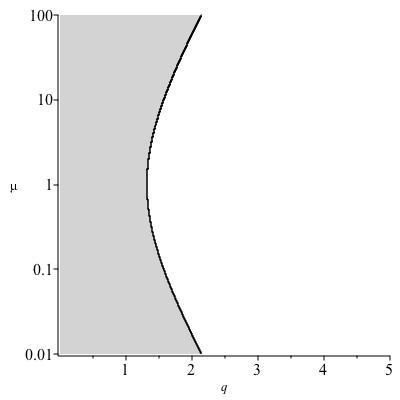}
\begin{minipage}{0.8\textwidth}
\caption{The shaded region corresponds to stable elliptic RE on the Lobachevsky plane, as a function of the distance $q$ (horizontally) and  mass ratio $\mu=\mu_1/\mu_2$ (vertically). The curve delimiting the two regions is $q=q^*(\mu)$ as defined in the statement of Theorem~\ref{thm:stability on L^2}. Note that the vertical $\mu$-axis on the figure above has a logarithmic scale.}
\label{F:Bif-diag-ellipticRE}
\end{minipage}
\end{figure}

\subsection*{Proof of Proposition\,\ref{P:Signature-L2}}

{\it (i)}: Fix $M\neq 0$ and  consider the restriction of the reduced  system to the symplectic leaf $\mathcal{M}_{-M^2}$ defined by $C( \boldsymbol  m , \boldsymbol  m ) =-M^2$. 
Denote by $H_{-M^2}$ the restriction of the reduced Hamiltonian~\eqref{eq:Hamiltonian on L2} to $\mathcal{M}_{-M^2}$. The RE with 
$C=-M^2$ correspond to critical points of $H_{-M^2}$. 

By the time-reversibility of the problem it is sufficient to consider the stability of hyperbolic RE for which $m_y>0$. 
We introduce coordinates\footnote{In fact $(\alpha,q,z,p)$ as defined  are Darboux coordinates that generalize the Andoyer variables on $T^*SO(3)$.}  $(\alpha,q,z,p)$
on  $\mathcal{M}_{-M^2}$ by setting:
$$
m_x=z, \quad m_y=\sqrt{M^2-z^2}\cosh\alpha, \quad m_z=\sqrt{M^2-z^2}\sinh\alpha.
$$
In view of~\eqref{eq:Hamiltonian on L2} and~\eqref{eq:grav-potential-L2},  we have
\begin{equation*}
\begin{split}
H_{-M^2}(\alpha,q,z,p)&=\dfrac{1}{2}\Bigg[(1+\mu)p^2-2pz+\dfrac{\cosh 2q-\cosh 2(q-\alpha)-2\mu \sinh^2\alpha}{2\sinh^2q }z^2 \\
&\qquad \qquad -\dfrac{1}{\sinh^2 q}\Big(\sinh 2q-\dfrac{M^2}{2}\big(\cosh (2(q-\alpha)) -1 +\mu(\cosh 2\alpha -1)\big)\Big)\Bigg].
\end{split}
\end{equation*}

From the proof of Theorem~\ref{thm:existence on L2}, we know that the critical points of $H_{-M^2}$ occur   at points where $z=p=0$, and
$q$ and $\alpha$ are such that \eqref{eq:CofMH} and \eqref{eq:MoHypH} hold. With our assumptions on the potential,   \eqref{eq:MoHypH}  simplifies to 
\begin{equation}
\label{eq:M2-L2-hyper}
M^2=\frac{\sinh q}{\sinh \alpha ( \mu \sinh \alpha \cosh q - \sinh(q-\alpha)}.
\end{equation}
The Hessian matrix of $ H_{-M^2}$  at such points is given by 
\begin{equation*}
D^2H_{-M^2} (\alpha,q,0,0) =:{\bf N} = \begin{pmatrix}
{\bf N}^{(1)} & 0 \\
0 & {\bf N}^{(2)} 
\end{pmatrix},
\end{equation*}
where ${\bf N}^{(1)}$ and ${\bf N}^{(2)}$ are symmetric $2\times 2$ matrices. The  entries of ${\bf N}^{(1)}$ may be written as
\begin{equation}
\label{eq:N1Hyper}
\begin{split}
&{\bf N}^{(1)}_{11}=\frac{M^2(\cosh (2(q -\alpha))+\mu\cosh 2\alpha) }{\sinh^2q}, \qquad {\bf N}^{(1)}_{22}= \dfrac{M^2(1+\mu)\sinh^2\alpha }{\sinh^4 q}, \\
&{\bf N}^{(1)}_{12}={\bf N}^{(1)}_{21}=- \dfrac{M^2(-\sinh q  \cosh 2\alpha +(1+\mu)\cosh q \sinh 2\alpha)}{\sinh^3 q} ,
\end{split}
\end{equation}
where we have used \eqref{eq:M2-L2-hyper} to simplify ${\bf N}^{(1)}_{22}$. On the other hand
\begin{equation}
\label{eq:N2Hyper}
{\bf N}^{(2)}=
\begin{pmatrix}
\dfrac{- \cosh(2(q-\alpha)) + \cosh 2q  -2\mu \sinh^2\alpha  }{\displaystyle\cosh 2 q -1 } & -1 \\
-1 & 1+\mu
\end{pmatrix}.
\end{equation}
A long but straightforward calculation  using \eqref{eq:CofMH} shows that 
\begin{equation*}
\det ({\bf N}^{(1)})=-\frac{M^4\sinh^2(q-\alpha)}{\cosh^2\alpha \sinh^4\alpha}\left ( 4 \cosh^2 \alpha \cosh^2 (q-\alpha) -1 \right ),
\end{equation*}
that is clearly negative. Hence ${\bf N}^{(1)}$ has one positive and one negative eigenvalue. On the other hand, by a calculation that uses again 
 \eqref{eq:CofMH} we obtain
\begin{equation*}
\det ({\bf N}^{(2)})=\frac{\sinh^2(q-\alpha)}{\cosh^2\alpha}>0.
\end{equation*}
Considering that ${\bf N}^{(2)}_{22}>0$ we conclude that ${\bf N}^{(2)}$ is positive definite. Therefore, the signature of ${\bf N}$
is $(+++-)$ as stated.

\noindent {\it (ii)} and \noindent {\it (iii)}: We proceed in an analogous fashion. Let $H_{M^2}$ be the restriction of the reduced Hamiltonian~\eqref{eq:Hamiltonian on L2} to
the symplectic leaf $\mathcal{M}_{M^2}$ defined by $C=M^2$. By the time-reversibility of the problem it is sufficient to consider the stability of hyperbolic RE for which $m_z>0$. Introduce coordinates\footnote{As before, $(\alpha,q,z,p)$ 
  are Darboux coordinates that generalize the Andoyer variables on $T^*SO(3)$.}   $(\alpha,q,z,p)$ on  $\mathcal{M}_{M^2}$ by:
$$
m_x=z, \quad m_y=\sqrt{M^2+z^2}\sinh\alpha, \quad m_z=\sqrt{M^2+z^2}\cosh\alpha.
$$
Then, in view of~\eqref{eq:Hamiltonian on L2} and~\eqref{eq:grav-potential-L2},  we have
\begin{equation*}
\begin{split}
H_{M^2}(\alpha,q,z,p)&=\dfrac{1}{2}\Bigg[(1+\mu)p^2-2pz+\dfrac{\cosh 2q+\cosh 2(q-\alpha)+\mu (1+\cosh 2\alpha)}{2\sinh^2q }z^2 \\
&\qquad \qquad -\dfrac{1}{\sinh^2 q}\Big(\sinh 2q-\dfrac{M^2}{2}\big(1+ \cosh (2(q-\alpha))  +\mu(1+ \cosh 2\alpha )\big)\Big)\Bigg].
\end{split}
\end{equation*}
From the proof of Theorem~\ref{thm:existence on L2}, we know that the critical points of $H_{M^2}$ occur   at points where $z=p=0$, and
$q$ and $\alpha$ are such that \eqref{eq:CofMH} and \eqref{eq:M0EH} hold. With our assumptions on the potential,   \eqref{eq:M0EH}  simplifies to 
\begin{equation}
\label{eq:M2-L2-ellip}
M^2=\frac{\sinh q}{\cosh \alpha ( \mu \cosh \alpha \cosh q +\cosh(q-\alpha)}.
\end{equation}
The Hessian matrix of $ H_{M^2}$  at such points has again block diagonal form 
\begin{equation*}
D^2H_{M^2} (\alpha,q,0,0) =:{\bf L} = \begin{pmatrix}
{\bf L}^{(1)} & 0 \\
0 & {\bf L}^{(2)} 
\end{pmatrix}.
\end{equation*}
The  elements of the matrices ${\bf L}^{(i)}$, $i=1,2$, coincide respectively with those of ${\bf N}^{(i)}$  given by \eqref{eq:N1Hyper} and  \eqref{eq:N2Hyper}
except for the entries $ {\bf L}^{(1)}_{22}$ and $ {\bf L}^{(2)}_{11}$ that may be written as 
\begin{equation*}
{\bf L}^{(1)}_{22}=\dfrac{M^2(1+\mu)\cosh^2\alpha }{\sinh^4 q}, \qquad  {\bf L}^{(2)}_{11}=\frac{\cosh \alpha (\cosh (2q-\alpha) + \mu \cosh \alpha)}{\sinh^2 q}.
\end{equation*}
The simplification of  $ {\bf L}^{(1)}_{22}$ given above is obtained with the help of \eqref{eq:M2-L2-ellip}. 

We have 
\begin{equation*}
\det ({\bf L}^{(2)}) = \frac{\mu^2\cosh^2 \alpha + 2\mu \cosh \alpha \cosh(q-\alpha)+ \cosh^2(q-\alpha)}{\sinh^2q}
\end{equation*}
which is clearly positive. Considering that $ {\bf L}^{(2)}_{22}>0$ it follows that ${\bf L}^{(2)}$ is positive definite. 

Next, given that ${\bf L}^{(1)}_{22}>0$ it follows that ${\bf L}^{(1)}$ has one positive eigenvalue and the signature of the other one coincides with the
signature of its determinant. On the other hand, using  \eqref{eq:CofMH} to eliminate $\mu$, we can factorize
\begin{equation*}
\det ({\bf L}^{(1)}) =\frac{4M^4\cosh^2\alpha \cosh^2(q-\alpha)}{\sinh^4q \sinh^22\alpha}f(q,\alpha),
\end{equation*}
with 
\begin{equation*}
f(q,\alpha):=1-4\sinh^2\alpha \sinh^2(q-\alpha).
\end{equation*}
The signature of $f$ along the RE can be determined by using  \eqref{eq:CofMH} to write 
\begin{equation}
\label{eq:q as fcn alpha hyper}
q=q(\alpha)=\alpha + \frac{1}{2}\mathrm{ arcsinh} (\mu( \sinh 2\alpha )).
\end{equation}
This expression leads to the identity $4\sinh^2(q-\alpha)=2\sqrt{1+\mu^2\sinh^2\alpha}-1$. Setting  $u=\sinh^2\alpha$ and using this identity  
 allows one to write
\begin{equation*}
f(q(\alpha),\alpha)=\tilde f(u):=2u\left ( 1-\sqrt{1+4\mu^2(1+u)u} \right ) +1.
\end{equation*}
It is now elementary to check that $\tilde f$ is decreasing for $u>0$, and satisfies $\tilde f(0)=1$ and $\lim_{u\to \infty} \tilde f(u)=-\infty$. Therefore, there exists a unique $u^*>0$ such that $\tilde f(u)$ is positive  for $0<u<u^*$ and  
 negative for $u^*<u$. Considering  that $u$ is an  increasing function of $\alpha$, 
 we conclude that $f(q(\alpha),\alpha)$ is positive for $0<\alpha <\alpha^*$ and negative for $\alpha^*<\alpha$ where $\alpha^*$ satisfies $\tilde f(\sinh^2\alpha^*)=0$.
The proof of items  {\it (ii)} and \noindent {\it (iii)} in the proposition is completed by noting that \eqref{eq:q as fcn alpha hyper} defines $q$ as an increasing function of $\alpha$.

It only remains to show that $M^2$ restricted to the branch of elliptic RE has a maximum at $q=q^*$. To show this we parametrize the branch by $\alpha>0$ using 
\eqref{eq:q as fcn alpha hyper}. Differentiating \eqref{eq:CofMH} implicitly with respect to $\alpha$ leads to 
\begin{equation*}
\frac{dq}{d\alpha} =\frac{\sinh(2(q-\alpha))\cosh 2 \alpha + \cosh(2(q-\alpha))\sinh 2\alpha}{\cosh(2(q-\alpha))\sinh 2\alpha}.
\end{equation*}
Starting from \eqref{eq:M2-L2-ellip}, using the above expression, and then  \eqref{eq:CofMH} to eliminate $\mu$, one may simplify
\begin{equation*}
\frac{dM^2}{d\alpha}=\frac{f(q,\alpha)}{\cosh^2\alpha \cosh^2 (q-\alpha)\cosh(2(q-\alpha))}.
\end{equation*}
Therefore, $dM^2/d\alpha$ same signature as $f(q,\alpha)$ (i.e. the same signature as $\det  ({\bf L}^{(1)})$ and $\det  ({\bf L})$). The same is true about 
$dM^2/dq$ since  \eqref{eq:q as fcn alpha hyper} defines $q$ as an increasing function of $\alpha$.
\hfill $\Box$

\subsection{Topology of the energy-momentum level surfaces}

Denote by $\mathcal{Z} \cong \R^4 \times \R^+$ the reduced space of the system with global coordinates $\boldsymbol  m,p,q$. 
In this section we consider the topology of the fibres of the energy-momentum map
$(C,H):\mathcal{Z}\to \R^2$. 
We continue working with the gravitational potential \eqref{eq:grav-potential-L2} under the assumption that $\mu_1=1$ and $G\mu_1\mu_2=1$. In particular note that 
\begin{equation*}
\lim_{q\to \infty}U_{\mathrm{ grav}}(q)=-1.
\end{equation*}

We shall prove that the  energy-momentum fibre over the point $(c_0,h_0)$ is compact only for $c_0>0$ and  $h_0< -1$, and that in this case it is homeomorphic to the disjoint union of two 3-spheres. 
We start by noticing that the energy-momentum map
$(C,H):\mathcal{Z}\to \R^2$ is not proper. It has the following properties:

\begin{proposition}
Let $\boldsymbol  \zeta_n=(\boldsymbol  m_n,p_n,q_n)$ be a sequence of points in $\mathcal{Z}$ and suppose that $(C,H)(\boldsymbol  \zeta_n)$ converges to $(c_0,h_0)\in \R^2$. Then
\begin{enumerate}
\item If $q_n\to \infty$ then $h_0\geqslant -1$.
\item If $q_n\to 0$ then $c_0\leqslant 0$.
\item If $\varepsilon \leqslant q_n \leqslant \frac{1}{\varepsilon}$ for some $\varepsilon>0$, then $\boldsymbol  m_n$ and $p_n$ are bounded.
\end{enumerate}
\end{proposition}
\begin{proof}
\begin{enumerate}
\item Since the kinetic energy is positive, we have $H(\boldsymbol  \zeta_n)\geqslant U_{\mathrm{grav}}(q_n)$, and the result follows by letting $n\to \infty$.
\item Suppose $C({\boldsymbol m}_n)=-(m_x)_n^2-(m_y)_n^2+ (m_z)_n^2\to c_0>0$. Then, for $n$ large enough, $(m_z)^2_n\geqslant \dfrac{c_0}{2}$ and we have
\begin{equation*}
H(\boldsymbol  \zeta_n)\geqslant  \frac{c_0\mu}{4\mu_1 \sinh^2 q_n } +U_{\mathrm{grav}}(q_n)= \frac{c_0\mu -4\mu_1 \cosh q_n \sinh q_n}{4\mu_1 \sinh^2 q_n},
\end{equation*}
that grows without bound as $q_n\to 0$. This contradicts our hypothesis that $H(\boldsymbol  \zeta_n)$ converges to $h_0$.
\item For  $\varepsilon \leqslant q \leqslant \frac{1}{\varepsilon}$ it is possible to bound  $H$ from below by a constant, positive definite
quadratic form on $\boldsymbol m_n$ and $p_n$ with constant coefficients, plus a constant value (the minimum of $U$ for $\varepsilon \leqslant q \leqslant \frac{1}{\varepsilon}$).
Hence, the only way in which $H$ can remain bounded if $q$ is bounded is if
$\boldsymbol m_n$ and $p_n$ are also  bounded.
\end{enumerate}
\end{proof}

Based on these observations we describe the topology of  the energy-momentum fibre over the point $(c_0,h_0)$ with $c_0>0$.  
\begin{theorem}\label{thm:EM fibres}
Assume $c_0>0$. The energy-momentum fibre over the point $(c_0,h_0)$ is:
\begin{itemize}
\item empty , 2 points, or homeomorphic to the disjoint union of two 3-spheres if $h_0< -1$,
\item for $h_0>-1$ and $(c_0,h_0)$ inside the small cusp region in Fig.\,\ref{F:DiagramL2}, it is homeomorphic to the disjoint union of two 3-spheres and two 3-dimensional open balls,
\item for $h_0>-1$ and $(c_0,h_0)$ outside the small cusp region in Fig.\,\ref{F:DiagramL2}, it is homeomorphic to the disjoint union of two 3-dimensional open balls.
\end{itemize}
\end{theorem}

On the other hand, all of the fibres having $c_0<0$ are unbounded but we were unable to  determine more information about their topology.  
The transition from $c_0>0$ to $c_0<0$ is surprisingly complicated.

\noindent\textbf{Discussion/proof:} We refer to Fig.\,\ref{F:DiagramL2}.   If $h_0<-1$ and $c_0$ is sufficiently large then the fibre is empty. Now assume $0<c_0<c_*$ where $c_*>0$ is the value of the Casimir where the saddle-node bifurcation takes place. For $h_0<<-1$ the fibre is empty, and there is a local minimum of $H$ on the level set of $C$ (since the signature of the critical point is $(+,+,+,+)$), and in fact two local minima related by the time-reversal symmetry.  In the region where this local minimum is less than $-1$, the minima are in fact global minima, and the fibre consists of just two isolated points.   These critical points are non-degenerate so by elementary Morse theory the $h_0$-level sets of $H$ on $C^{-1}(c_0)$ is the union of two 3-spheres, for $H_{\textrm{min}}<h_0<-1$. 

As $h_0$ becomes larger than $-1$, there are two open sets arising `from infinity' (large $q$). This gives a contribution to the fibre of two open balls. As $(c_0,h_0)$ leaves the small cusp region by increasing $h_0$, there is a critical point of $H$ of signature $(+,+,+,-)$. This corresponds to the spheres meeting the interior of the open ball and then coalescing, giving rise to just two open balls.

\newpage
\section{Relative equilibria for the  2-body problem on the sphere}

Consider two masses $\mu_1$ and $\mu_2$ on the unit sphere, interacting via an attracting conservative force with potential energy $U(q)$, where $q$ is the angular separation of the masses, where by attracting we mean $U'(q)>0$ for all $q\in(0,\pi)$.  

\subsection{Classification of relative equilibria}
\label{SS:Classification S2}

The symmetry group of this problem is $SO(3)$, and every 1-parameter subgroup consists of rotations about a fixed axis. Thus,  all of the relative equilibria of the problem for $\omega\neq0$ 
are periodic solutions, and the 2 masses simultaneously rotate about a fixed axis of rotation at a steady angular speed $\omega$.  Given any configuration with $q\in(0,\pi)$ there is a uniquely defined shortest geodesic between the two masses.  Any point lying on this arc is said to lie \emph{between} the masses.

Recall from the introduction that in counting relative equilibria, we identify any relative equilibria that merely differ by a symmetry (including time-reversing symmetry).

\begin{theorem}\label{thm:existence on S^2}
In the 2-body problem on the sphere, governed by an attractive force, the set of relative equilibria (RE) depends on whether the masses are equal or not as follows. In every case, the axis of rotation  lies between the masses in the sense described above. 
\begin{description}
\item[$\mu_2\neq\mu_1$:] For each $q\in(0,\pi)$, $q\neq\pi/2$ there is a unique RE where the masses are separated by an angle $q$. The axis of rotation  subtends an angle $\theta_j\in(0,\pi/2)$ with the mass $\mu_j$ (so $q=\theta_1+\theta_2$) which are related by
\begin{equation}
\label{eq:RE-cond1}
\mu_1\sin (2\theta_1) = \mu_2\sin(2\theta_2).
\end{equation}
We call these \emph{acute} and \emph{obtuse RE}, accordingly as $q<\pi/2$ or $q>\pi/2$.  There is no RE for $q=\pi/2$.  In the acute RE, the larger mass is closer to the axis of rotation, while in the obtuse RE, the smaller mass is closer.  See Fig.\,\ref{fig:acute-and-obtuse-S2}.

\item[$\mu_2=\mu_1$:] In this case there are two classes of RE, isosceles and right-angled (see Fig.\,\ref{fig:isosceles+RA}):
\begin{enumerate}
\item Given any $q\in(0,\pi)$, $q\neq\pi/2$, there is a unique RE where the masses are separated by an angle $q$. In this case the axis of rotation passes through the sphere midway between the masses; these we call  \emph{isosceles RE}. 
\item Given any $\theta\in(0,\pi/2)$ there is a unique RE with angular separation $q=\pi/2$, called a \emph{right-angled RE}, where $\theta$ is the smaller of the angles between the axis of rotation and the masses. 
\end{enumerate}
Note that when $q=\pi/2$ and $\theta=\pi/4$ these two families meet, giving just one RE. 
\end{description}
For all the RE, the speed of rotation $\omega$ is given by 
\begin{equation}\label{eq:omegaSphere}
\omega^2=\zeta^{-1}U'(q),
\end{equation}
where $\zeta=\frac12\mu_1\sin(2\theta_1)=\frac12\mu_2\sin(2\theta_2)$ and of course $q=\theta_1+\theta_2$.
\end{theorem}

Previous existence results for the RE of the problem for the gravitational potential \eqref{eq:gravity on S2} were given before in \cite{3}, \cite{PCh}. The theorem above
completes their classification for arbitrary attractive potentials.

\begin{figure}
\centering
\includegraphics{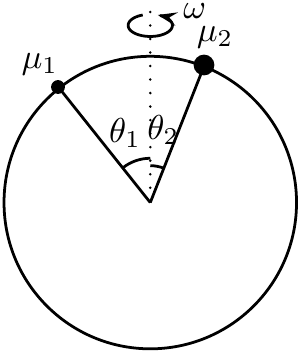} 
	\qquad\qquad\qquad 
\includegraphics{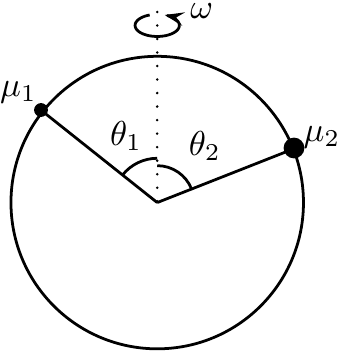}
\begin{minipage}{0.8\textwidth}
\caption{Acute and obtuse relative equilibria for distinct masses (here shown for $\mu=0.7$, and $q=\pi/3$ and $2\pi/3$ respectively).
}
\label{fig:acute-and-obtuse-S2}
\end{minipage}
\end{figure}

\begin{figure}
\centering
\includegraphics{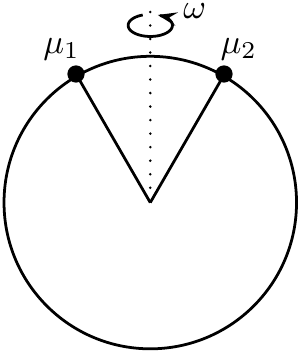} 
	\qquad\qquad\qquad 
\includegraphics{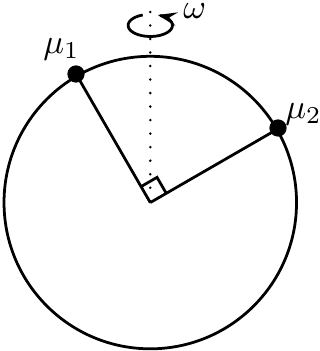}
\begin{minipage}{0.8\textwidth}
\caption{Isosceles and right-angled relative equilibria respectively, for a pair of equal masses.
}
\label{fig:isosceles+RA}
\end{minipage}
\end{figure}

\subsection*{Proof of Theorem\,\ref{thm:existence on S^2}}

Relative equilibria are equilibrium points  of the reduced system \eqref{eq5.2} so they correspond to solutions of  the following equations:
\begin{subequations}\label{eq:RE}
\begin{align}
\label{eq:RE-p}
\frac{\partial H}{\partial p}&=0, \\
\label{eq:RE-m}
\frac{\partial H}{\partial {\boldsymbol m}}\times {\boldsymbol m} & ={\bf 0}, \\
\label{eq:RE-th}
\frac{\partial H}{\partial q} & =0,
\end{align}
\end{subequations}
where the reduced Hamiltonian $H$ is given in \eqref{eq:Hamiltonian on sphere}.

The condition \eqref{eq:RE-p} is equivalent to
\begin{equation}
\label{eq:p}
p=-\frac{m_x}{1+\mu}.
\end{equation}
Substituting \eqref{eq:p} into  the last two components of \eqref{eq:RE-m} yields two possibilities:
\begin{enumerate}
\item $m_y=m_z=0$. In this case \eqref{eq:RE-th} implies  $U'(q)=0$. Since the potential is attractive there is no solution of this form.
\item $m_x=0$. We focus on this case in what follows. Note that \eqref{eq:p} implies that $p=0$.
\end{enumerate}

Introduce polar coordinates for non-zero points in the  $m_y$--$m_z$ plane by putting 
\begin{equation}
\label{eq:PolarSphere}
m_z=M_0\cos \alpha,  \qquad m_y= M_0\sin \alpha,
\end{equation}
for $\alpha \in [0, 2\pi)$ and $M_0>0$ (we interpret $\alpha$ in terms of the configuration geometry below). The first component of  \eqref{eq:RE-m} and \eqref{eq:RE-th} may be rewritten in these coordinates as:
\begin{subequations}\label{eq:RE-cond2}
\begin{align}
\label{eq:RE-Mom}
&\mu\sin (2\alpha) - \sin(2(q -\alpha))=0, \\
 \label{eq:M0cond}
&M_0^2 =
\frac{\mu_1 \sin^3q\, U'(q)}
{F_\mu(q,\alpha)}.
\end{align}
\end{subequations}
where 
\begin{equation}
\label{eq:def of F}
F_\mu(q,\alpha):=\cos \alpha (\mu \cos q \cos \alpha +  \cos (q- \alpha)).
\end{equation}
It follows from \eqref{eq:M0cond} that $M_0$ has a real value if and only if $F_\mu(q,\alpha)>0$.  Note also that if $(q,\alpha,M_0)$ is a solution to \eqref{eq:RE-cond2}, then so is $(q,\alpha+\pi,M_0)$.   This corresponds to changing the sign of $\boldsymbol m$, which is the time-reversing symmetry described in the introduction.  Since we do not count such pairs of solutions separately, we restrict attention from now on to $\alpha\in[0,\pi)$.  

The analysis of the solutions of \eqref{eq:RE-cond2} depends on whether or not the masses are equal.

\bigskip

\noindent\textbf{Equal masses:} In this case $\mu=1$ and it follows from \eqref{eq:RE-Mom} that either 
$q=\pi/2$ or $q=2\alpha (\bmod\pi)$.  

Firstly, if $q=\pi/2$ then $F_1(\pi/2,\alpha) = \cos\alpha\sin\alpha$ and this is strictly positive if and only if $\alpha\in(0,\pi/2)$.  

Now suppose $q=2\alpha (\bmod\pi)$.  Then again, $F_1(q,\alpha) >0$ if and only if $\alpha\in(0,\pi/2)$; in particular, if $\alpha>\pi/2$ then $q=2\alpha-\pi$ and $F_1(q,\alpha) =-\cos^2\alpha(\cos2\alpha+1)<0$.

For equal masses, the equations therefore have a solution if and only if $\alpha\in(0,\pi/2)$, with either $q=\pi/2$ or $q=2\alpha$.  We return to this after the analogous discussion for distinct masses. 

\bigskip

\noindent\textbf{Distinct masses:}  Without loss of generality we suppose $0<\mu<1$. The relation between $q$ and $\alpha$ from \eqref{eq:RE-Mom} is shown in Fig.\,\ref{fig:q-alpha-S2}. For each value of $\alpha\in(0,\pi)$ it is clear that there are precisely two values of $q\in[0,\pi)$ satisfying  equation \eqref{eq:RE-cond2}, namely
\begin{equation} \label{eq:q+-}
\begin{split}
q_-(\alpha)&:=\alpha + \frac{1}{2}\arcsin (\mu( \sin 2\alpha)) \bmod\pi, \\
q_+(\alpha)&:=\alpha +\frac{\pi}{2} - \frac{1}{2}\arcsin (\mu( \sin 2\alpha))\bmod \pi.
\end{split}
\end{equation}

\begin{figure}
\centering
\includegraphics{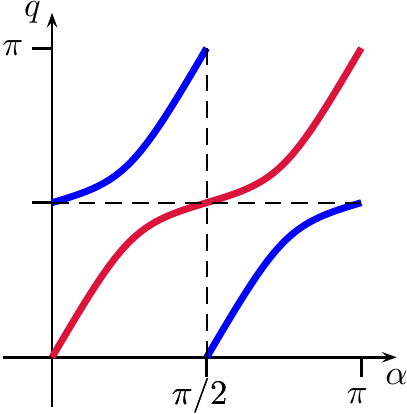}
\begin{minipage}{0.8\textwidth}
\caption{The relation between $q$ and $\alpha$ given by \eqref{eq:RE-Mom} for $\mu\neq1$ (here   $\mu=0.7$). The red curve represents $q=q_-(\alpha)$ while the blue one represents $q=q_+(\alpha)$.}
\label{fig:q-alpha-S2}
\end{minipage}
\end{figure}

The following lemma is proved at the end of the section.

\begin{lemma} \label{lemma:distinct masses}
Let $g_{\pm}:[0,\pi )\to \R$ be the functions defined by
\begin{equation*}
g_{\pm}(\alpha)=F_\mu(q_\pm(\alpha),\alpha),
\end{equation*}
Then, for any value of $\mu\in(0,1)$,  each function $g_{\pm}$ is continuous, and moreover it is strictly positive if and only if $\alpha\in (0,\pi/2)$, and in this case $q_\pm(\alpha)>\alpha$. 
\end{lemma}

If $q=\pi/2$ then $\alpha=0 \bmod \pi/2$ and $g_\pm(\alpha)=0$. Consequently, $M_0^2$ in \eqref{eq:M0cond} is undefined, showing there are no solutions with $q=\pi/2$. 

Now for any $q\in(0,\pi/2)$, there is just one value of $\alpha\in(0,\pi/2)$ (as required by the lemma above) which is related by $q=q_-(\alpha)$.  These correspond to acute relative equilibria. 

Similarly, for any $q\in(\pi/2,\pi)$, there is again a single value of $\alpha\in(0,\pi/2)$ satisfying \eqref{eq:RE-Mom}, but now $q=q_+(\alpha)$. These are obtuse relative equilibria. 

\bigskip

There now remains to relate $(q,\alpha)$ to the angles $\theta_1,\theta_2$ of the theorem.  Note that for all the solutions described above (with $\alpha\in(0,\pi/2)$) we have $q>\alpha$ (for equal masses this is clear, for distinct masses it follows from the lemma).  From \eqref{eq:PolarSphere} and the definition of the moving frame it follows that $\alpha=0$ corresponds to the axis containing the mass $\mu_1$, and then $\alpha$ increases in the direction towards $\mu_2$.  Thus $\alpha$ represents the angle between $\mu_1$ and the axis of rotation (the axis containing $\boldsymbol m$), so $\alpha=\theta_1$. Moreover, since $q>\alpha$, the axis of rotation lies between the masses, and the angle between the axis of rotation and $\mu_2$ is $\theta_2=q-\alpha$. This shows that \eqref{eq:RE-cond1} is equivalent to \eqref{eq:RE-Mom}.

In view of the above discussion, and of the inequality
\begin{equation*}
q_-(\alpha)<2\alpha < q_+(\alpha), \qquad \alpha\in(0,\pi/2),
\end{equation*}
that is easily established from \eqref{eq:q+-},  it follows that $\theta_2<\theta_1$ for acute RE, with the opposite inequality, $\theta_1<\theta_2$, holding for obtuse RE. This proves the claim
about which mass is closer to the axis of rotation.

Finally, the equation \eqref{eq:omegaSphere} for $\omega$ is obtained by starting with  \eqref{eq1.9e} and using  $m_x=0$, \eqref{eq:PolarSphere} and
 \eqref{eq:RE-cond2} to simplify the resulting expression.

\begin{proof}[Proof of Lemma\,\ref{lemma:distinct masses}.]
We begin by showing the continuity of $g_\pm$. Since $F_\mu(q,\alpha)$ is continuous in $(q,\alpha)$ the only candidate for discontinuities of $g_\pm$ is where $q_\pm(\alpha)$ is discontinuous in $\alpha$, and this can only occur where $q_\pm=0,\pi$. For $q_-$ this only occurs for $\alpha=0,\pi$ so does not give rise  to a discontinuity. For $q_+(\alpha)$, which is increasing on each subinterval $(0,\pi/2)$ and  $(\pi/2,\pi)$, this can only occur when $\alpha=\pi/2$, at which point $F_\mu(q,\pi/2)=0$ and the discontinuity of $q_+$ has no effect on $g_+$.

Next, for $\alpha\in (0,\pi/2)$ it is clear that $0< q_-(\alpha)<\pi/2$ and $0<  q_-(\alpha)-\alpha<\pi/2$. Therefore  $g_-(\alpha)>0$ since all terms in its expression are positive.

Now let $\alpha\in (0,\pi/2)$ and let us show that $g_+(\alpha)>0$.  Let $\beta=\frac{1}{2}\arcsin(\mu\sin 2\alpha)$, then $\beta\in (0,\pi/4)$ and we have
$\sin^2\beta =\frac{1}{2}-\frac{1}{2}\sqrt{1-\mu^2\sin^22\alpha}$. Hence,
 $$0<\sin^4\beta =\sin^2\beta- \mu^2 \sin^2\alpha\cos^2\alpha,$$ which implies 
\begin{equation}
\label{eq:auxprooflemma}
 \mu  \sin \alpha\cos \alpha< \sin \beta .
\end{equation}
Given that $q_+(\alpha)=\alpha -\beta +\pi/2$ we have
\begin{equation*}
\frac{g_+(\alpha)}{\cos \alpha}= \mu \sin \beta \cos^2 \alpha - \mu \cos \beta \sin \alpha \cos \alpha + \sin \beta \geq - \mu \cos \beta \sin \alpha \cos \alpha + \sin \beta,
\end{equation*}
and one can easily prove that the quantity on the right is positive using \eqref{eq:auxprooflemma}.

Now note that for $\alpha\in (0,\pi/2)$ one has
\begin{equation*}
q_{\pm}(\pi/2 +\alpha) = \pi -q_{\pm}(\pi/2-\alpha).
\end{equation*}
Using this relation in the definition of $g_{\pm}$ shows that $g_{\pm}$ is ``odd" with respect to $\alpha=\pi/2$. Namely,
\begin{equation*}
g_\pm(\pi/2 +\alpha)=-g_\pm(\pi/2 -\alpha).
\end{equation*}
This last expression  shows  that  $g_{\pm}(\alpha)\leq 0$ for values of $\alpha\in [0,\pi)$ that do not lie on $(0,\pi/2)$.

The final inequality in the statement of the lemma follows immediately from \eqref{eq:RE-Mom} given that $\alpha\in(0,\pi/2)$. 
\end{proof}

\subsection{Stability of the relative equilibria on the sphere}
As for the Lobachevsky plane, the stability of the relative equilibria found above depends on the precise form of the potential, and in this section we restrict attention to the gravitational potential \eqref{eq:gravity on S2}.
Similar to our treatment in $L^2$, in our analysis we assume that the constants $\mu_1$ and $G  \mu_1\mu_2$
 equal one. In this way, the gravitational potential $U(q)=-\cot(q)$, and the Hamiltonian depends on the parameters
of the problem only through the mass ratio $\mu$ that we will continue to assume to be $0<\mu<1$.

The following theorem synthesises the results of the  linear stability analysis of the problem.
In its statement  \emph{elliptic} means that the linearisation of the reduced equations restricted to the symplectic leaf only has non-zero,
purely imaginary eigenvalues.

\begin{theorem} \label{thm:stability on S^2}
For the relative equilibria described in Theorem\,\ref{thm:existence on S^2}, the linear stability analysis of the RE 
depends on whether the masses are equal or not as follows.
\begin{description}
\item[$\mu_2\neq\mu_1$:] \begin{enumerate}
\item All acute RE are elliptic. 
\item There exists a critical obtuse angle $q^*\in (\pi/2,\pi)$,
 which depends on the mass ratio $\mu$, such that obtuse RE are elliptic for $\pi/2<q<q^*$, and are unstable  for $q^*<q<\pi$. \newline
The critical angle $q^*$ is given by $q^*=q_+(\alpha^*)$ where $\alpha^*$ is defined implicitly as the unique solution in $(0,\pi/2)$ of the equation
 $$
\cos 2\alpha=2\sin^2\alpha \sqrt{1-\mu^2\sin^2 2\alpha }.
$$
 Moreover, along the branch of obtuse RE,  the momentum $M^2$ has a minimum at  $q^*$.
 
\end{enumerate}
\item[$\mu_2=\mu_1$:] 

 \begin{enumerate}
\item All right-angled RE with $\theta \neq \pi/4$ are elliptic. 
\item All isosceles RE subtending an acute angle $q\in (0,\pi/2)$ are elliptic.
\item All isosceles RE subtending an obtuse angle $q\in (\pi/2, \pi)$ are unstable.
\end{enumerate}
\end{description}
\end{theorem}

Contrary to the case in $L^2$ the Hamiltonian function cannot be used as a Lyapunov function to guarantee the nonlinear stability of the elliptic
RE of the problem. 
This is due to the non-definiteness of the Hamiltonian at these points. More precisely we have
\begin{proposition}
\label{P:Signature}
The Hessian matrix of the reduced Hamiltonian  (restricted to the corresponding symplectic leaf)
at the RE of the problem described in Theorem\,\ref{thm:existence on S^2}, has the following signature:
\begin{description}
\item[$\mu_1\neq \mu_2$.]
\begin{enumerate}
\item Acute RE have signature $(++--)$.
\item Obtuse RE with $\pi/2<q<q^*$ have signature $(++--)$.
\item Obtuse RE with $q^*<q<\pi$ have signature $(+++-)$.
\end{enumerate}
\item[$\mu_1\neq \mu_2$.]
\begin{enumerate}
\item Right-angled RE which are not isosceles have signature $(++--)$.
\item Isosceles  RE subtending an acute angle have signature $(++--)$.
\item Isosceles  RE subtending an obtuse angle have signature   $(+++-)$.
\end{enumerate}
\end{description}
\end{proposition}

A proof of the results of Theorem~\ref{thm:stability on S^2} and Proposition~\ref{P:Signature} is given in Section~\ref{SS:LinearS2} below.\footnote{An alternative proof
 is given in \cite{MRO-2017} by working on a symplectic slice of the unreduced system.}
The results  are conveniently summarised in the  energy-momentum diagram of the system given in Fig.\,\ref{F:EM-S2-DiffMass}.
\begin{figure}[h]
\centering
\subfigure[The case of different masses, \,$\mu_2=2\mu_1$.
Notice the change in signature at the cusp of the obtuse relative equilibria corresponding to $q=q^*$, where $M^2$ has a minimum. ]{\includegraphics[width=6cm]{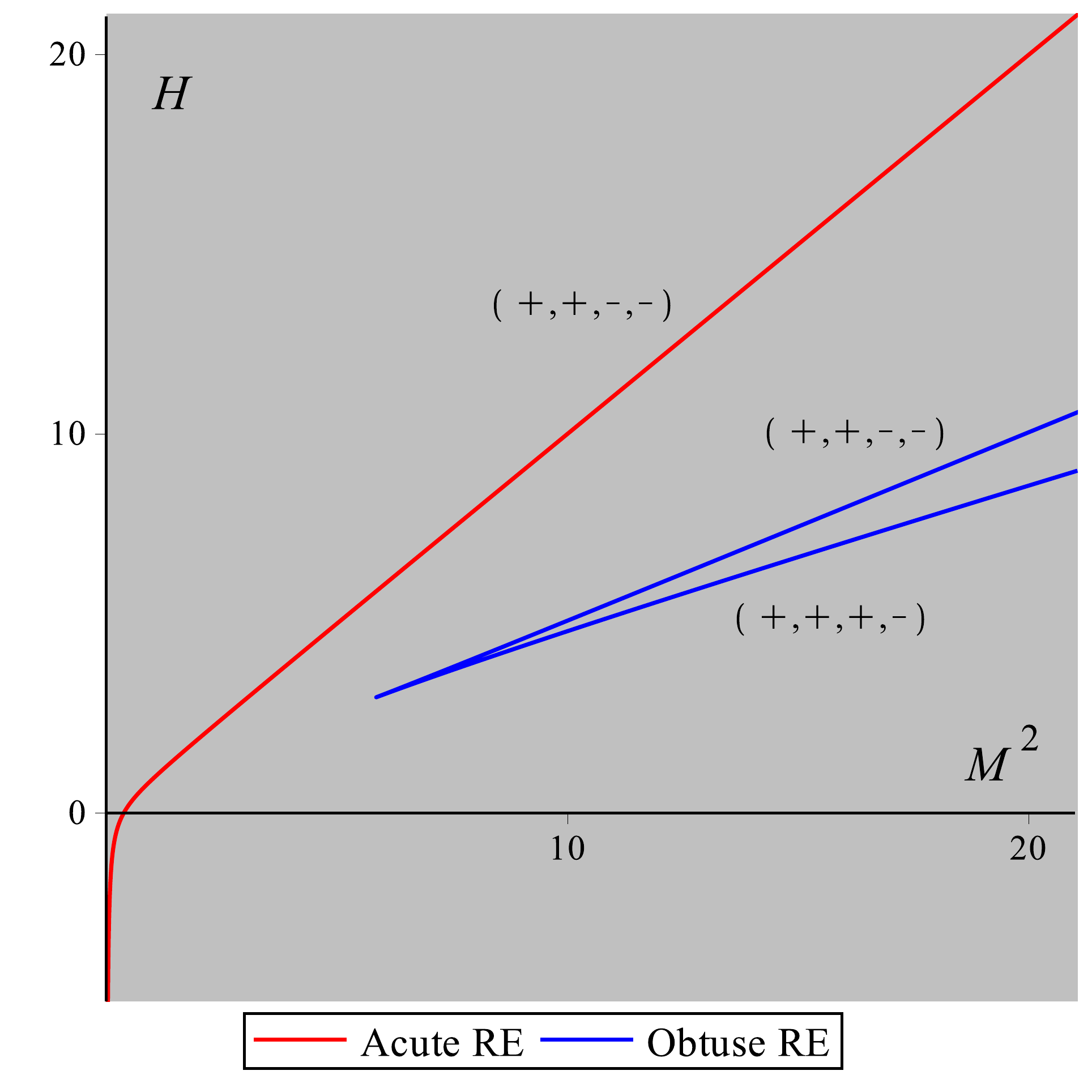}} \qquad
\subfigure[The case of equal masses, \,$\mu_1= \mu_2$.
There is a change of   signature passing from acute  to obtuse
isosceles RE when the two families of RE intersect. ]{\includegraphics[width=6cm]{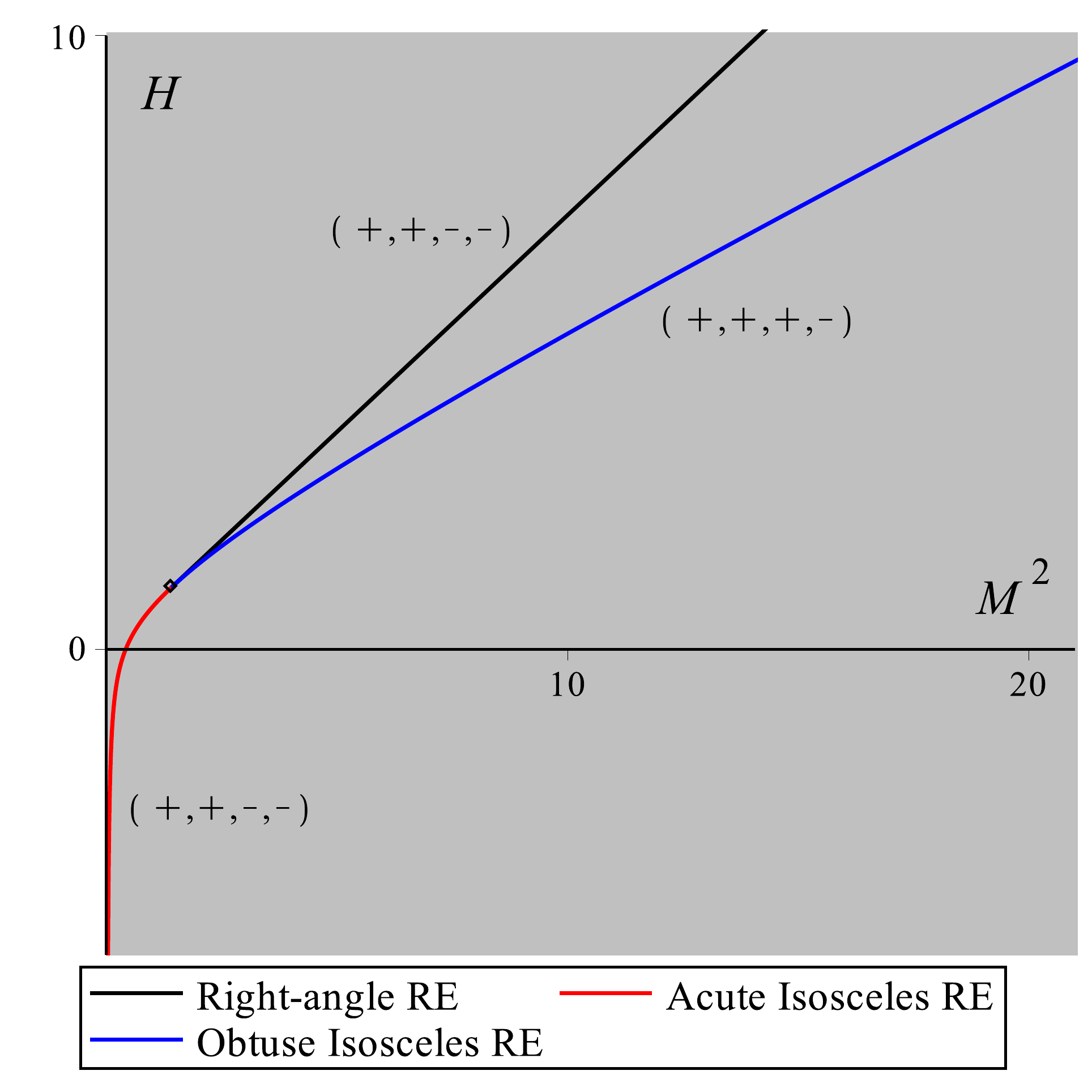}}
\begin{minipage}{0.8\textwidth}
\caption{Energy-Momentum bifurcation diagram of relative equilibria  on the sphere for the gravitational potential. 
The shaded area on the $M^2$-$H$ plane show all possible values of $(M^2,H)$. All RE with signature $(++--)$ are elliptic.   }
\label{F:EM-S2-DiffMass}
\end{minipage}
\end{figure}

\begin{remark}
Considering that the gravitational potential $U$ satisfies  $\lim_{q\to 0}U(q)=-\infty$, it is easy to construct a sequence $q_n, p_n$ with
$q_n\to 0$, and $p_n\to \infty$, such that $H$ evaluated at $(\boldsymbol m, q, p)=(0,M_0,0,q_n,p_n)$ converges to any arbitrary value $h_0\in \R$.
This shows that the energy-momentum map $(M^2,H)$ is not proper and  that all of its fibres are non-compact.
\end{remark}

The results presented above are insufficient to show the nonlinear stability of the RE of the problem.
By nonlinear stability we mean the  stability in the sense of Lyapunov in the reduced space of the corresponding equilibrium. Since the restriction
of the reduced system to the symplectic leaves defines a 2-degree of freedom Hamiltonian system, a nonlinear stability
analysis may be performed using Birkhoff normal forms and applying KAM theory. We do this in Section~\ref{KAM}, but
we restrict our attention to the acute RE. These are parametrised by the mass ratio $\mu\in (0,1)$ (an `external' parameter) and the
arc $q\in (0,\pi/2)$ (an `internal' parameter). In our treatment in Section~\ref{KAM} we
  give numerical evidence of the validity of the following (which we are calling a `theorem' although we have not carried out an analytic proof).
  
\begin{theorem} \label{thm:KAM}
Acute RE are nonlinearly stable in an open dense set of the parameter space $(\mu,q)\in (0,1)\times (0,\pi/2)$.
\end{theorem}

\subsubsection{Linear analysis}
\label{SS:LinearS2}

Fix $M\neq 0$ and  consider the restriction of the reduced  system to the symplectic leaf $\mathcal{M}_{M^2}$ defined by $C( \boldsymbol  m , \boldsymbol  m ) =M^2$. 
Denote by $H_{M^2}$ the restriction of the reduced Hamiltonian~\eqref{eq:Hamiltonian on sphere} to $\mathcal{M}_{M^2}$. The RE with 
$C=M^2$ correspond to critical points of $H_{M^2}$.

Introduce canonical (Andoyer) coordinates on $\mathcal{M}_{M^2}$ by setting:
$$
m_x=z, \quad m_y=\sqrt{M^2-z^2}\sin\alpha, \quad m_z=\sqrt{M^2-z^2}\cos\alpha.
$$
Then $(\alpha,q,z,p)$ are Darboux coordinates  on $\mathcal{M}_{M^2}$ and the restriction of the reduced equations of motion to $\mathcal{M}_{M^2}$
takes the canonical form
$$
 \dot{\alpha}=\dfrac{\partial  H_{M^2}}{\partial z}, \quad \dot{q}=\dfrac{\partial  H_{M^2}}{\partial p}, \quad   \dot{z}=-\dfrac{\partial  H_{M^2}}{\partial \alpha}, \quad
 \dot{p}=-\dfrac{\partial  H_{M^2}}{\partial q}.
$$
where, in view of~\eqref{eq:Hamiltonian on sphere},  we have
\begin{equation*}
\begin{split}
H_{M^2}(\alpha,q,z,p)&=\dfrac{1}{2}\Bigg[(1+\mu)p^2+2pz-\dfrac{\cos 2q+\cos 2(q-\alpha)+\mu(1+\cos 2\alpha)}{2\sin^2 q}z^2 \\
&\qquad \qquad -\dfrac{1}{\sin^2 q}\Big(\sin 2q-\dfrac{M^2}{2}\big(1+\cos 2(q-\alpha)+\mu(1+\cos 2\alpha)\big)\Big)\Bigg].
\end{split}
\end{equation*}

From our discussion in  Section~\ref{SS:Classification S2}, we know that the critical points of $H_{M^2}$ occur   at points where $z=p=0$, and
$q$ and $\alpha$ are such that \eqref{eq:RE-cond2} holds. Note that under our assumptions on the potential $U$, equation~\eqref{eq:M0cond} simplifies to
\begin{equation}
\label{eq:Mom-grav-S2}
M^2=\frac{\sin q}{\cos \alpha (\cos (q-\alpha)+ \mu \cos \alpha \cos q)}.
\end{equation}
The Hessian matrix of $ H_{M^2}$  along the equilibrium points is given by 
\begin{equation*}
D^2H_{M^2} (\alpha,q,0,0) =:{\bf N} = \begin{pmatrix}
{\bf N}^{(1)} & 0 \\
0 & {\bf N}^{(2)} 
\end{pmatrix},
\end{equation*}
where ${\bf N}^{(1)}$ and ${\bf N}^{(2)}$ are symmetric $2\times 2$ matrices. The  entries of ${\bf N}^{(1)}$ may be written as
\begin{equation}
\label{112}
\begin{split}
&{\bf N}^{(1)}_{11}=-\frac{M^2(\cos (2(q -\alpha))+\mu\cos 2\alpha) }{\sin^2q}, \qquad {\bf N}^{(1)}_{22}= \dfrac{M^2(1+\mu)\cos^2\alpha }{\sin^4 q}, \\
&{\bf N}^{(1)}_{12}={\bf N}^{(1)}_{21}= \dfrac{M^2(-\sin q  \cos 2\alpha +(1+\mu)\cos q \sin 2\alpha)}{\sin^3 q} ,
\end{split}
\end{equation}
where we have used \eqref{eq:Mom-grav-S2} to simplify ${\bf N}^{(1)}_{22}$. On the other hand
\begin{equation}
{\bf N}^{(2)}=
\begin{pmatrix}
-\dfrac{\cos\alpha \big(\cos (2q-\alpha)  +\mu \cos\alpha)  }{\displaystyle\sin^2 q } & 1 \\
1 & 1+\mu
\end{pmatrix}.
\end{equation}

\begin{lemma}
The matrix $\bf N$ is indefinite for all RE of the problem. It has at least 2 positive eigenvalues and at least 1 negative eigenvalue.
\end{lemma}
\begin{proof}
That $\bf N$ has at least two positive eigenvalues follows immediately from its block diagonal form and the inequalities
$$
{\bf N}_{22}^{(1)}>0, \qquad {\bf N}_{22}^{(2)}>0.
$$
To complete the proof we  show that $ {\bf N}^{(2)}$ has a negative eigenvalue. This follows directly from the expression
$$
\det ( {\bf N}^{(2)})=-\frac{\cos^2(q-\alpha)}{\sin^2\alpha}<0
$$
that is obtained by eliminating $\mu$ using~\eqref{eq:RE-Mom}.
\end{proof}

Let us now introduce the quantity
\begin{equation}
\label{eq:defa}
a:={ \bf N}_{11}^{(1)}{ \bf N}_{11}^{(2)}+2{ \bf N}_{12}^{(1)}{ \bf N}_{12}^{(2)}+{ \bf N}_{22}^{(1)}{ \bf N}_{22}^{(2)}
=\frac{M^2(1+\cos(2(q-\alpha))}{\sin^2q\sin^2\alpha},
\end{equation}
that will be relevant in what follows. To obtain the simplified expression  on the right, one needs to 
use \eqref{eq:Mom-grav-S2} and then \eqref{eq:RE-Mom} to eliminate $\mu$.  Using again these conditions, one may show by a lengthy but straightforward
calculation that
\begin{equation}
\label{eq:determinantS2}
b:=\det{ \bf N} =\frac{a^2f(q,\alpha)}{4}.
\end{equation}
where
\begin{equation}
\label{eq:f-forS2}
f(q,\alpha):=1-4\sin^2\alpha\sin^2(q-\alpha ).
\end{equation}

\begin{lemma}
\label{L:DetS2}
The determinant of the  matrix $\bf N$ along the RE is
\begin{description}
\item[$\mu_1\neq \mu_2$]
\begin{enumerate}
\item positive for all acute RE and for the obtuse RE with $\pi/2<q<q^*$,
\item negative for obtuse RE with $q^*<q<\pi$.
\end{enumerate}
In items (i) and (ii) above, $q^*$ is defined in the statement of Theorem~\ref{thm:stability on S^2}.
\item[$\mu_1= \mu_2$]
\begin{enumerate}
\item positive for right angled RE with $\alpha\neq \pi/4$ and for isosceles RE subtending acute angle,
\item negative for isosceles RE subtending an obtuse angle.
\end{enumerate}
\end{description}
\end{lemma}

\begin{proof}
It is clear from \eqref{eq:determinantS2} that the sign of $\det{ \bf N}$ coincides with the sign of $f(q,\alpha)$. We analyse the latter along the
 different RE of the problem.
\begin{description}
\item[$\mu_1\neq \mu_2$.]
\begin{enumerate}
\item For acute RE we have $q=q_-(\alpha)$ and therefore
$$
f=1-2\sin^2\alpha \Big(1-\sqrt{1-\mu^2 \sin^2 2\alpha }\Big).
$$
Since $\mu^2\sin^2 2\alpha <\sin^2 2\alpha $ for all $\mu\in (0, 1)$, we find that
$$
f>1-2\sin^2\alpha (1-|\cos 2\alpha |)=
\begin{cases}
1-4\sin^4\alpha \quad 0<\alpha \leqslant\pi/4, \\
\cos^2 2\alpha  \quad \pi/4\leqslant\alpha <\pi/2.
\end{cases}
$$
The above function is everywhere greater than zero except at the point $\alpha=\pi/4$, but, as can be verified,
$$
f\big|_{\alpha=\pi/4}=\sqrt{1-\mu^2}>0.
$$

\item For obtuse RE we have $q=q_+(\alpha)$ and the corresponding expression for $f$ is $f=1-2\sin^2\alpha  \left ( 1+\sqrt{1-\mu^2\sin^22\alpha} \right )$.
This is a strictly decreasing function of $\alpha$ on the interval $(0,\pi/2)$, that is positive for $0<\alpha< \alpha_*$ and negative for $\alpha_*< \alpha<\pi/2$, with
 $\alpha_*$ defined in the statement of  Theorem~\ref{thm:stability on S^2}
\end{enumerate}

\item[$\mu_1= \mu_2$.] Suppose now that $\mu=1$.
\begin{enumerate}
\item Along the right-angled RE we have $q=\pi/2$ and 
\begin{equation*}
f=(2\cos^2\alpha -1)^2.
\end{equation*}

\item Along the isosceles RE we have $q=2\alpha$ and 
\begin{equation*}
f=(3-2\cos^2\alpha)(2\cos^2\alpha -1).
\end{equation*}
\end{enumerate}
The statement in the case $\mu_1=\mu_2$ follows immediately from the above two equalities.
\end{description}
\end{proof}

Combining the above lemmas gives a proof of Proposition~\ref{P:Signature}.
The statements about instability in Theorem~\ref{thm:stability on S^2} follow from this proposition. Now we show the elliptic nature of 
the other RE.

\begin{lemma}
All RE of the problem having signature $(++--)$ are elliptic.
\end{lemma}
\begin{proof}
The matrix of the linearisation of the system at RE is 
\begin{equation*}
{\bf A}=\begin{pmatrix} 0 & I\\-I &0 \end{pmatrix} {\bf N}=\begin{pmatrix} 0 &  {\bf N}^{(2)}\\- {\bf N}^{(1)} &0 \end{pmatrix} 
\end{equation*}
with characteristic polynomial
$$
P(\lambda)=\lambda^4+a \lambda^2+b,
$$
where $a$ is given by~\eqref{eq:defa} and $b=\det{\bf N}$. Since both $a$ and $b$ are positive, the ellipticity  condition is equivalent to $R_1>0$ where
\begin{equation}
\label{eq:defR1}
R_1:=\frac{1}{4}a^2- b.
\end{equation}
However, using \eqref{eq:determinantS2} we may write
$$
R_1=a^2\sin^2\alpha \sin^2(q-\alpha ),
$$
which is positive along all RE considered in the statement of the lemma.
\end{proof}

Finally, to complete the  proof of Theorem~\ref{thm:stability on S^2},  we show that along the branch of obtuse RE, the momentum $M^2$ has a minimum at $q=q^*$.
We parametrise the branch by $\alpha\in (0,\pi/2)$ by writing $q=q_+(\alpha)$. Differentiating~\eqref{eq:RE-Mom} implicitly with respect to $\alpha$ leads to
\begin{equation*}
\frac{dq}{d\alpha}=1+\tan(2(q-\alpha))\cot (2\alpha).
\end{equation*}
Now, by differentiating~\eqref{eq:Mom-grav-S2} with respect to $\alpha$,  using the above expression for $\frac{dq}{d\alpha}$,  and~\eqref{eq:RE-Mom} to eliminate $\mu$, we find
\begin{equation*}
\frac{dM^2}{d\alpha}=\frac{f(\alpha,q)}{\cos^2(q-\alpha)\cos(2(q-\alpha))\cos^2\alpha},
\end{equation*}
where $f(\alpha,q)$ is given by \eqref{eq:f-forS2}.  It is easy to show, from the definition of $q_+(\alpha)$ in \eqref{eq:q+-},  that $\pi/2<2( q_+(\alpha)-\alpha) <\pi$,
for  $\alpha\in (0,\pi/2)$. Therefore,
$\cos(2(q-\alpha))<0$ in the above formula,  and  $\frac{dM^2}{d\alpha}$ has the opposite sign of $f(q,\alpha)$. Hence, in view of 
  \eqref{eq:determinantS2}, we conclude that    $\frac{dM^2}{d\alpha}$ has the opposite sign of $\det{ \bf N}$.
The result now follows from Lemma~\ref{L:DetS2} and the fact that $\frac{dq_+}{d\alpha}>0$.

\subsubsection{Nonlinear stability of acute  relative equilibria in the case of different masses}
\label{KAM}

In this section we give numerical evidence for the validity of Theorem~\ref{thm:KAM}.

Consider an elliptic  RE of the problem that projects to an (isolated) equilibrium point on $\mathcal{M}_{M^2}$. 
By Theorem~\ref{thm:stability on S^2} the eigenvalues of the linearized system are purely imaginary
$$
\lambda_1=i \Omega_1, \quad \lambda_2=-i\Omega_1, \quad \lambda_3=i\Omega_2, \quad \lambda_4=-i \Omega_2, \quad 0<\Omega_1<\Omega_2.
$$
Our investigation of its 
nonlinear stability will proceed by checking that the following two conditions
are satisfied:
\begin{enumerate}
\item[{$1^\circ$.}] there  are no   second or third-order resonances:
$$
\Omega_2\ne 2\Omega_1, \quad \Omega_2\ne 3\Omega_1.
$$
Under this condition  one may put the Hamiltonian (restricted to  $\mathcal{M}_{M^2}$) in Birkhoff normal form
\begin{equation}
\label{4012}
H=\frac{1}{2}\sum\limits^2_{j=1} \alpha_j I_j +\frac{1}{4}\sum^2_{j, k=1}\beta_{jk}I_j I_k+O_5, \quad I_j=x_j^2+y_j^2, \quad |\alpha_j|=\Omega_j.
\end{equation}
Here, $x_j$ and $y_j$~are suitable canonical coordinates on a neighbourhood of the equilibrium on  $\mathcal{M}_{M^2}$ (i.e., $\{x_j, y_k\}=\delta_{jk}$)  with the equilibrium located
  at $x_j=y_j=0$, $\beta_{jk}$ are constants, and $O_5$ denotes
a power series containing terms of order no less that 5 in $x_j,\, y_j$.

If conditions $1^\circ$ and $2^\circ$ are satisfied, a sufficient condition for nonlinear stability (under perturbations within  $\mathcal{M}_{M^2}$) may be given in terms of the nonlinear terms  in
\eqref{4012}. Specifically, one requires that

\item[{$2^\circ.$}] the {\it Arnold determinant} is different from zero
$$
D:=\det
\begin{pmatrix}
\beta_{11} & \beta_{12} & \alpha_1 \\
\beta_{12} & \beta_{22} & \alpha_2 \\
\alpha_1 & \alpha_2 & 0
\end{pmatrix}=2\beta_{12}\alpha_1\alpha_2-\beta_{11}\alpha^2_2-\beta_{22}\alpha_1^2 \ne 0.
$$
\end{enumerate}
This nonlinear condition allows one to apply the  KAM theorem in such a way that the  invariant tori act as boundaries for the flow on each constant energy surface,  leading to Lyapunov stability
of the equilibrium (see e.g.  \cite{Moser_1968}, \S 35 in~\cite{Siegel-Moser}, or   Section 13 in~\cite{MeyerHall} for proofs and details).

\begin{remark}
If the Arnold determinant $D=0$, one may still obtain sufficient conditions for stability by considering higher order terms in the normal form expansion \eqref{4012} (see e.g.~\cite{MeyerHall}).
On the other hand, the presence of second or third-order resonances may lead to instability.
We shall not consider any of these possibilities here.
\end{remark}

\begin{remark}
We emphasise that the above analysis ensures nonlinear stability of RE only  with respect to perturbations on the initial conditions that
lie on the  momentum surface  $\mathcal{M}_{M^2}$.
\end{remark}

\subsubsection*{Resonances}
Recall that the characteristic polynomial of the linearised system is
$$
P(\lambda)=\lambda^4+a \lambda^2+b.
$$
Condition  $1^\circ$, which requires that there are no second or third-order resonances, is written in terms of the coefficients $a$ and $b$  as
\begin{equation}
\label{cond2-alt}
\begin{gathered}
R_2:=\frac{4}{25}a^2-b\ne 0, \quad R_3:=\frac{9}{100}a^2-b\ne0.
\end{gathered}
\end{equation}
 Fig.\,\ref{fig21} below illustrates the plane $a$-$b$ of coefficients of the characteristic polynomial. The curves $\sigma_1$ and $\sigma_2$ respectively correspond to
the values of $(a,b)$ attained at the acute and obtuse RE of the problem for the fixed value of $\mu=0.95$.
 These RE are conveniently parametrised by $\alpha \in (0,\pi/2)$ by setting $q=q_-(\alpha)$ (acute) and $q=q_+(\alpha)$ (obtuse). 
 The figure also illustrates the parabolae corresponding
to the zero loci of $R_1$ defined by \eqref{eq:defR1}, and of the second and third order resonance polynomials $R_2$ and $R_3$ defined in \eqref{cond2-alt}.
 \begin{figure}[!ht]
	\centering
	\includegraphics{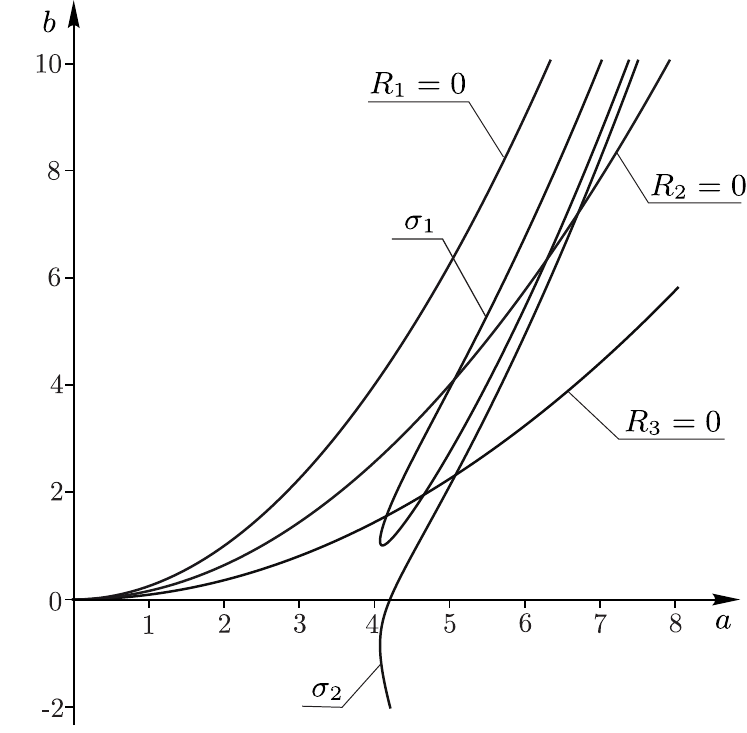}
\begin{minipage}{0.8\textwidth}
	\caption{Curves $\sigma_1$, $\sigma_2$,  corresponding to acute and obtuse RE 
	on the coefficient plane $(a, b)$ for $\mu=0.95$.}
	\label{fig21}
\end{minipage}
\end{figure}

An analytic investigation of condition~\eqref{cond2-alt} for a general $0<\mu<1$ involves very heavy calculations so we present numerical results. 
We restrict our attention to the acute RE that may be
parametrised  by $\alpha \in (0,\pi/2)$ by putting $q=q_-(\alpha)$.  Our  results are then presented in terms of the parameters $(\mu,\alpha )\in (0,1)\times (0,\pi/2)$.

One can express $R_2$ and $R_3$ as functions of $(\mu,\alpha )$ by  substituting $q=q_-(\alpha)$  into  \eqref{eq:defa} and  \eqref{eq:determinantS2}.
 The zero loci of $R_2$ and $R_3$ on the $\alpha$-$\mu$-plane
are the two curves illustrated in Fig.\,\ref{fig10}.

\subsubsection*{Analysis of the Arnold determinant}

As for the resonance condition, we only present numerical results for our investigation of condition  $2^\circ$ for  the acute RE. By using  \eqref{eq:M0cond} and 
setting $q=q_-(\alpha)$,
we express $D=D(\mu,\alpha)$.

The zero locus of $D$  on the plane $\alpha$-$\mu$ consists of the  two curves illustrated in Fig.\,\ref{fig10} that provides numerical evidence for the
validity of Theorem~\ref{thm:KAM}.

\begin{figure}[h!]
\centering
\includegraphics[width=0.31\textwidth]{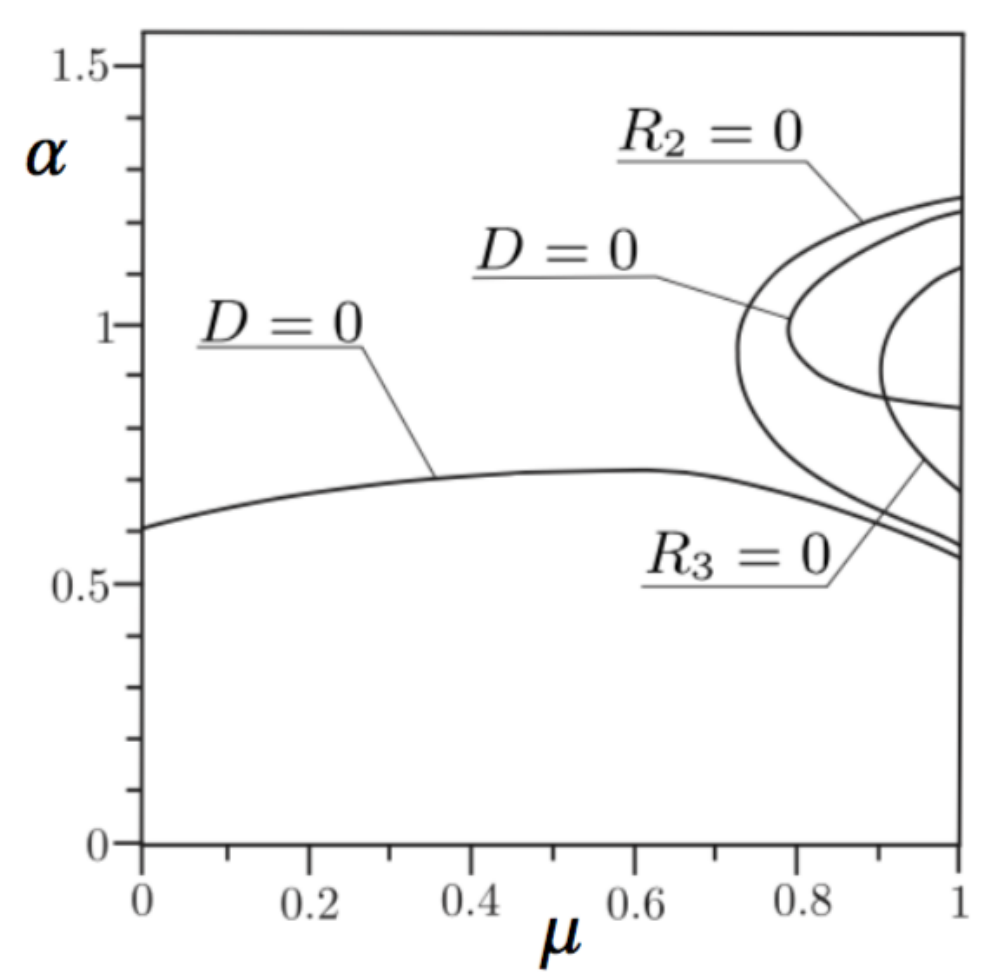}
\caption{{Curves on the plane $\mu$-$\alpha$ plane corresponding to RE with second and third order resonances (respectively $R_2=0$ and $R_3=0$) and where
	the Arnold determinant vanishes ($D=0$).}}
\label{fig10}
\end{figure}

\subsubsection{Open problems}
\addcontentsline{toc}{section}{Open problems}
To conclude, we point out a number of open problems concerning the stability of the  two-body problem in the sphere:
\begin{itemize}
\item[{--}] For different masses, investigate the nonlinear stability of acute RE
 for which there are  resonances $(R_2=0, \, R_3=0)$ and/or the Arnold determinant vanishes $(D=0)$.
\item[{--}] Again for different masses in $S^2$, investigate the nonlinear stability of obtuse RE that are elliptic.
\item[{--}] In the case of equal masses, investigate the nonlinear stability of acute-isosceles RE and right-angled RE.
\item[{--}] Classify and investigate the stability of all RE for the spatial two-body problem on $S^3$ and $L^3$.
\end{itemize}

\section*{Acknowledgements}

We are thankful to both reviewers and the associate editor for their remarks and criticisms which  led to an improvement of our paper.

We are grateful to Miguel Rodr\'iguez-Olmos for discussing his preliminary results of~\cite{MRO-2017} with us.
The authors express their gratitude to B.\,S.~Bardin and I.\,A.~Bizyaev for fruitful discussions and useful comments.

The research contribution of LGN and JM was made possible by a Newton Advanced Fellowship from the Royal Society, ref:~NA140017.

The work of AVB and ISM is supported by the Russian Foundation for Basic Research (project No.\,17-01-00846-a).
The research of AVB was also carried out within the framework of the state assignment of the Ministry of Education and Science of Russia.

\vskip 1cm
\setlength{\parindent}{0pt}

\emph{A.V.\,Borisov} \\
Udmurt State University, \\
ul. Universitetskaya 1, Izhevsk, 426034 Russia \\
\textit{and} \\
A.A.Blagonravov Mechanical Engineering Research Institute of RAS \\
ul. Bardina 4, Moscow, 117334 Russia \\
\texttt{borisov@rcd.ru}

\bigskip

\emph{L.C.~Garc\'ia-Naranjo} \\
Departamento de Matem\'aticas y Mec\'anica IIMAS-UNAM \\
Apdo. Postal: 20-726 Mexico City, 01000, Mexico\\
\texttt{luis@mym.iimas.unam.mx}

\bigskip

\emph{I.S.~Mamaev} \\ 
Institute of Mathematics and Mechanics of the Ural Branch of RAS \\
ul. S.Kovalevskoi 16, Ekaterinburg, 620990 Russia\\
\textit{and}\\
Izhevsk State Technical University \\
Studencheskaya 7, Izhevsk, 426069 Russia \\
\texttt{mamaev@rcd.ru}

\bigskip

\emph{J.~Montaldi} \\
School of Mathematics, University of Manchester \\
Manchester M13 9PL, UK\\
\texttt{j.montaldi@manchester.ac.uk}


\begin{thebibliography}{99}{}

\bibitem{5}
Bolsinov A.V., Borisov A.V., Mamaev I.S.,
Topology and stability of integrable systems.
Russian Mathematical Surveys, 2010, vol.65, no.2, pp.259--318.

\bibitem{6}
Bolsinov A.V., Borisov A.V., Mamaev I.S.,
The Bifurcation Analysis and the Conley Index in Mechanics.
Regular and Chaotic Dynamics, 2012, vol.17, no.5, pp.457--478.


\bibitem{DTT}
Borisov A.V., Mamaev I.S.,
Rigid body dynamics. Hamiltonian methods, integrability, chaos.
Moscow--Izhevsk: Institute of Computer Science, 2005, 576 p. (in Russian).
	

\bibitem{4}
Borisov A.V., Mamaev I.S.,
Reduction in the two-body problem on the Lobatchevsky plane.
Russian Journal of Nonlinear Dynamics, 2006, vol.2, no.3, pp.279--285 (in Russian).

\bibitem{7}
Borisov A.V., Mamaev I.S.,
Rigid Body Dynamics in NonEuclidean Spaces.
Russian Journal of Mathematical Physics,
2016, vol.23, no.4, pp.431--453.

\bibitem{BM_2006}
Borisov A.V., Mamaev I.S.,
The Restricted Two-Body Problem in Constant Curvature Spaces.
Celestial Mech. Dynam. Astronom., 2006, vol.96, no.1, pp.1--17.


\bibitem{1}
Borisov A.V., Mamaev I.S., Bizyaev I.A.,
The Spatial Problem of 2 Bodies on a Sphere.
Reduction and Stochasticity. Regular and Chaotic Dynamics, 2016, vol.21, no.5, pp.556--580.

\bibitem{3}
Borisov A.V., Mamaev I.S., Kilin A.A.,
Two-Body Problem on a Sphere: Reduction, Stochasticity, Periodic Orbits.
Regul. Chaotic Dyn., 2004, vol.9, no.3, pp.265--279.



\bibitem{Carinena} Cari\~nena J.F., M.F. Ra\~nada, M. Santander, Central potentials on spaces of constant curvature: the Kepler problem
on the two dimensional sphere $S^2$ and the hyperbolic plane $H^2$,  J. Math. Phys., 2005, vol.46, 052702.


\bibitem{Ch}
Chernoivan V.A., Mamaev I.S.,
The Restricted Two-Body Problem and the Kepler Problem in the Constant Curvature Spaces.
Regular and Chaotic Dynamics, 1999, vol.4, no.2, pp.112--124.


\bibitem{Diacu}
Diacu F.,
Relative Equilibria of the Curved $N$-Body Problem.
Paris: Atlantis, 2012.

\bibitem{DiacuPCh}  Diacu F., P\'erez-Chavela E.,  Reyes J.G., An intrinsic approach in the curved n-body problem. The negative case.
J. Differential Equations, 2012, vol.252, pp.4529--4562.

\bibitem{LGN2016} Garc\'ia-Naranjo L.C., Marrero J.C.,
P\'erez-Chavela E. and Rodr\'iguez-Olmos M., 
Classification and stability of relative equilibria for the
two-body problem in the hyperbolic space of dimension 2.
J. of Differential Equations, 2016, vol.260, pp.6375--6404.

\bibitem{Iversen} Iversen, B., Hyperbolic Geometry. LMS Student Texts, no.\ 25. 1992.

\bibitem{Kilin}
Kilin A.A.,
Libration Points in Spaces $S^2$ and $L^2$.
Regular and Chaotic Dynamics, 1999, vol.4, no.1, pp.91--103.

\bibitem{kozlov} Kozlov, V.V., Harin, A.O., Kepler's problem in constant curvature spaces, {\it Cel. Mech. Dynam. Astro.},  (1992),  {\bf 54}, 393--399,


\bibitem{Markeev_1978}
Markeev A.P.,
Libration points in celestial mechanics and cosmodynamics.
M.: Nauka, 1978, 312 p. (in Russian).


\bibitem{Marsden}
Marsden, J.E, 
Lectures on Mechanics. 
C.U.P. (1992).


\bibitem{MeyerHall} Meyer K.R., Hall G.R, and Offin, D.,  
Introduction to Hamiltonian dynamical systems and the $N$-body problem.
Second edition. Applied Mathematical Sciences, {\bf 90}. Springer, New York, 2009.


\bibitem{Mo-Peyresq}
Montaldi, J., 
Relative equilibria and conserved quantities in symmetric Hamiltonian systems. 
In: Peyresq Lectures in Nonlinear Phenomena, 
World Scientific, 2000. 

\bibitem{MN-G}
Montaldi, J., Nava-Gaxiola, C., 
Point vortices on the hyperbolic plane.
J.\ Math.\ Phys., vol.~55 (2014), 102702, 14 pp.

\bibitem{Mont}
Montanelli H.,
Computing Hyperbolic Choreographies.
Regular and Chaotic Dynamics, 2016, vol.~21, no.5, pp.523--531.

\bibitem{Moser_1968}
Moser J.,
Lectures on Hamiltonian systems.
American Mathematical Soc., 1968, no.81.

\bibitem{PCh} P\'erez-Chavela, E., Reyes-Victoria, J. G., An intrinsic approach in the curved n-body problem. The positive curvature case. 
Trans. Amer. Math. Soc., vol.~364 (2012), 3805--3827. 
 
\bibitem{MRO-2017}
Rodr\'iguez-Olmos M., 
Relative equilibria for the two-body problem on $S^2$. In preparation.

\bibitem{2}
Shchepetilov A.V.,
Two-Body Problem on Spaces of Constant Curvature: 1.Dependence of the   Hamiltonian on the Symmetry Group and the Reduction of the Classical System.
Theoret.\ and Math.\ Phys.,  2000, vol.124, no.2, pp.1068--1081.

\bibitem{Sh}
Shchepetilov A.V.,
Calculus and Mechanics on Two-Point Homogenous Riemannian Spaces.
Lect.\ Notes Phys., vol.707, Berlin:Springer, 2006.


 \bibitem{Siegel-Moser} Siegel C.L. and  Moser J.K.,  
Lectures on celestial mechanics. 
Translated from the German by C. I. Kalme. Reprint of the 1971 translation. Classics in Mathematics. Springer-Verlag, Berlin, 1995.


\end{thebibliography}
\end{document}